\documentclass[11pt]{article}
\usepackage{style}
\usepackage{framed}
\usepackage[parfill]{parskip}

\allowdisplaybreaks

\title{Inapproximability of VC Dimension and Littlestone's Dimension}

\author{
Pasin Manurangsi\thanks{Email: \texttt{pasin@berkeley.edu}.} \\
UC Berkeley
\and
Aviad Rubinstein\thanks{Email: \texttt{aviad@berkeley.edu}.} \\
UC Berkeley
}

\date{\today}

\begin{document}

\maketitle

\begin{abstract}
We study the complexity of computing the VC Dimension and Littlestone's Dimension.
Given an explicit description of a finite universe and a concept class (a binary matrix whose $(x,C)$-th entry is $1$ iff element $x$ belongs to concept $C$),
both can be computed exactly in quasi-polynomial time ($n^{O(\log n)}$).
Assuming the randomized Exponential Time Hypothesis (ETH), we prove nearly matching lower bounds on the running time, that hold even for {\em approximation} algorithms.
\end{abstract}

\section{Introduction}
% !TeX root = main.tex

%Learning arbitrarily complex concepts is impossible. 
A common and essential assumption in learning theory is that the concepts we want to learn come from a nice, simple concept class, or (in the agnostic case) they can at least be approximated by a concept from a simple class.
When the concept class is sufficiently simple, there is hope for good (i.e. sample-efficient and low-error) learning algorithms. 

There are many different ways to measure the {\em simplicity} of a concept class. The most influential measure of simplicity is the VC Dimension, which captures learning in the PAC model. We also consider Littlestone's Dimension~\citep{Lit88}, which corresponds to minimizing mistakes in online learning (see Section~\ref{sec:prelim} for definitions).
When either dimension is small, there are algorithms that exploit the simplicity of the class, to obtain good learning guarantees.

Two decades ago, it was shown (under appropriate computational complexity assumptions) that neither dimension can be computed in polynomial time \citep{PY96, FL98}; and these impossibility results hold even in the most optimistic setting where the entire universe and concept class are given as explicit input (a binary matrix whose $(x, C)$-th entry is $1$ iff element $x$ belongs to concept $C$). 
The computational intractability of computing the (VC, Littlestone's) dimension of a concept class suggests that even in cases where a simple structure exists, it may be inaccessible to computationally bounded algorithms (see Discussion below).

In this work we extend the results of \citep{PY96, FL98} to show that the VC and Littlestone's Dimensions cannot even be {\em approximately} computed in polynomial time.
We don't quite prove that those problems are \NP-hard:
both dimensions can be computed (exactly) in quasi-polynomial ($n^{O(\log n)}$) time, 
hence it is very unlikely that either problem is \NP-hard. 
Nevertheless, assuming the randomized Exponential Time Hypothesis (ETH)%
\footnote{The randomized ETH (rETH) postulates that there is no $2^{o(n)}$-time Monte Carlo algorithms that solves $3\SAT$ on $n$ variables correctly with probability at least $2/3$ (i.e. $3\SAT \notin \BPTIME(2^{o(n)})$).} \citep{IPZ01, IP01},
% (see Section~\ref{sec:prelim} for definition).
we prove essentially tight quasi-polynomial lower bounds on the running time - that hold even against approximation algorithms.

\begin{theorem}[Hardness of Approximating VC Dimension]\label{thm:VC}
Assuming Randomized ETH, approximating VC Dimension to within a $(1/2 + o(1))$-factor requires $n^{\log^{1 - o(1)}n}$ time.
\end{theorem}

\begin{theorem}[Hardness of Approximating Littlestone's Dimension]\label{thm:LS}
There exists an absolute constant $\varepsilon > 0$ such that, assuming Randomized ETH, approximating Littlestone's Dimension to within a $(1 - \varepsilon)$-factor requires $n^{\log^{1 - o(1)}n}$ time.
\end{theorem}

\subsection{Discussion}
As we mentioned before, the computational intractability of computing the (VC, Littlestone's) dimension of a concept class suggests that even in cases where a simple structure exists, it may be inaccessible to computationally bounded algorithms.
We note however that it is not at all clear that any particular algorithmic applications are immediately intractable as a consequence of our results.

Consider for example the adversarial online learning zero-sum game corresponding to Littlestone's Dimension:
At each iteration, Nature presents the learner with an element from the universe;
the learner attempts to classify the element, and loses a point for every wrong classification;
at the end of the iteration, the correct (binary) classification is revealed.
The Littlestone's Dimension is equal to the worst case loss of the Learner before learning the exact concept.
(see Section~\ref{sec:prelim} for a more detailed definition.)

What can we learn from the fact that the Littlestone's Dimension is hard to compute?
The first observation is that there is no efficient learner that can commit to a concrete mistake bound.
But this does not rule out a computationally-efficient learner that plays optimal strategy and makes at most as many mistakes as the unbounded learner.
We can, however, conclude that Nature's task is computationally intractable!
Otherwise, we could efficiently construct an entire worst-case mistake tree (for a concept class ${\cal C}$, any mistake tree has at most $|{\cal C}|$ leaves, requiring $|{\cal C}|-1$ oracle calls to Nature).

On a philosophical level, we think it is interesting to understand the implications of an intractable, adversarial Nature. Perhaps this is another evidence that the mistake bound model is too pessimistic?

Also, the only algorithm we know for computing the optimal learner's decision requires computing the Littlestone's Dimension. We think that it is an interesting open question whether an approximately optimal computationally-efficient learner exists.

In addition, let us note that in the other direction, computing Littlestone's Dimension exactly implies an exactly optimal learner. However, since the learner has to compute Littlestone's Dimension many times, we have no evidence that an approximation algorithm for Littlestone's Dimension would imply any guarantee for the learner.

Finally, we remark that for either problem (VC or Littlestone's Dimension), we are not aware of any non-trivial approximation algorithms.

\subsection{Techniques}
The starting point of our reduction is the framework of ``birthday repetition''~\citep{AIM14}.
This framework has seen many variations in the last few years, but the high level approach is as follows: begin with a hard-to-approximate instance of a $2$CSP (such as $3$-\textsc{Color}), and partition the vertices into $\sqrt{n}$-tuples. 
On one hand, by the birthday paradox, even if the original graph is sparse, we expect each pair of random $\sqrt{n}$-tuples to share an edge; this is crucial for showing hardness of approximation in many applications.
On the other hand our reduction size is now approximately $N \approx 2^{\sqrt{n}}$ (there are $3^{\sqrt{n}}$ ways to color each $\sqrt{n}$-tuple), whereas by ETH solving $3$-\textsc{Color} requires approximately $T(n) \approx 2^n$ time, so solving the larger problem also takes at least $T(n) \approx N^{\log N}$ time.

\paragraph{VC Dimension}
The first challenge we have to overcome in order to adapt this framework to hardness of approximation of VC Dimension is that the number of concepts involved in shattering a subset $S$ is $2^{|S|}$. 
Therefore any inapproximability factor we prove on the size of the shattered set of elements, ``goes in the exponent'' of the size of the shattering set of concepts. Even a small constant factor gap in the VC Dimension requires proving a polynomial factor gap in the number of shattering concepts (obtaining polynomial gaps via ``birthday repetition'' for simpler problems is an interesting open problem~\citep{MR16, Man17}). Fortunately, having a large number of concepts is also an advantage: we use each concept to test a different set of $3$-\textsc{Color} constraints chosen independently at random; if the original instance is far from satisfied, the probability of passing all $2^{\Theta(|S|)}$ tests should now be doubly-exponentially small ($2^{-2^{\Theta(|S|)}}$)!
More concretely, we think of half of the elements in the shattered set as encoding an assignment, and the other half as encoding which tests to run on the assignments.
%It is crucial for our analysis that we keep those half separate, which prevents us from proving a stronger inapproximability factor (see Section~\ref{sub:VC-candidate} for more details).

\paragraph{Littlestone's Dimension}
Our starting point is the reduction for VC Dimension outlined in the previous paragraph.
While we haven't yet formally introduced Littlestone's Dimension, recall that it corresponds to an online learning model.
If the test-selection elements arrive before the assignment-encoding elements, the adversary can adaptively tailor his assignment to pass the specific test selected in the previous steps.
To overcome this obstacle, we introduce a special gadget that forces the assignment-encoding elements to arrive first; this makes the reduction to Littlestone's Dimension somewhat more involved. 
Note that there is a reduction by \citep{FL98} from VC Dimension to Littlestone's Dimension. 
Unfortunately, their reduction is not (approximately) gap-preserving, so we cannot use it directly to obtain Theorem~\ref{thm:LS} from Theorem~\ref{thm:VC}.

\subsection{Related Work}
The study of the computational complexity of the VC Dimension was initiated by Linial, Mansour, and Rivest \citep{LMR91}, who observed that it can be computed in quasi-polynomial time.
\citep{PY96} proved that it is complete for the class \LOGNP~which they define in the same paper.
\citep{FL98} reduced the problem of computing the VC dimension to that of computing Littlestone's Dimension, hence the latter is also \LOGNP-hard.
(It follows as a corollary of our Theorem~\ref{thm:VC} that, assuming ETH, solving any \LOGNP-hard problem requires quasi-polynomial time.)

Both problems were also studied in an implicit model, 
where the concept class is given in the form of a Boolean circuit that takes as input an element $x$ and a concept $c$ and returns $1$ iff $x \in c$.
Observe that in this model even computing whether either dimension is $0$ or not is already \NP-hard.
Schafer proved that the VC Dimension is $\Sigma_3^{\P}$-complete~\citep{Sch99}, while the Littlestone's Dimension is \PSPACE-complete~\citep{Sch00}.
\citep{MU02} proved that VC Dimension is $\Sigma_3^{\P}$-hard to approximate to within a factor of almost $2$; can be approximated to within a factor slightly better than $2$ in \AM; and is \AM-hard to approximate to within $n^{1-\varepsilon}$.

Another line of related work in the implicit model proves computational intractability of PAC learning (which corresponds to the VC Dimension). 
Such intractability has been proved either from cryptographic assumptions, e.g. \citep{KV94, Kharitonov93, Kharitonov95, FGKP06, KKMS08, KS09, Klivans16} or from average case assumptions, e.g. \citep{DS16, Daniely16}.
\citep{Blum94} showed a ``computational'' separation between PAC learning and online mistake bound (which correspond to the VC Dimension and Littlestone's Dimension, respectively):
 if one-way function exist, then there is a concept class that can be learned by a computationally-bounded learner in the PAC model, but not in the mistake-bound model.

Recently,~\citep{BFS16} introduced a generalization of VC Dimension which they call Partial VC Dimension, and proved that it is \NP-hard to approximate (even when given an explicit description of the universe and concept class).
%studied computationally-bounded learning in the related models of PAC learning (corresponding to VC Dimension) and online mistake bound (Littlestone's Dimension); he showed that if one-way functions exist, there is a concept class that can learned by a computationally-bounded learner only in the former model.

Our work is also related to many other quasi-polynomial lower bounds from recent years, which were also inspired by ``birthday repetition''; these include problems like Densest $k$-Subgraph~\citep{BKRW17, Man17}, Nash Equilibrium and related problems~\citep{BKW15, R15, BPR16, R16, BCKS16, DFS16} and Community Detection~\citep{R16-communities}.
It is interesting to note that so far ``birthday repetition'' has found very different applications, but they all share essentially the same quasi-polynomial {\em algorithm}: The bottleneck in those problem is a bilinear optimization problem $\max_{u,v} u^{\top}Av$, which we want to approximate to within a (small) constant additive factor. It suffices to find an $O(\log n)$-sparse sample $\hat v$ of the optimal $v^*$; the algorithm enumerates over all sparse $\hat v$'s \citep{LMM03, AGSS12, Barman15, CCDEHT15}.
In contrast, the problems we consider in this paper have completely different quasi-polynomial time algorithms: For VC Dimension, it suffices to simply enumerate over all $\log{|\cC|}$-tuples of elements (where $\cC$ denotes the concept class and $\log{|\cC|}$ is the trivial upper bound on the VC dimension)~\citep{LMR91}. Littlestone's Dimension can be computed in quasi-polynomial time via a recursive ``divide and conquer'' algorithm (See Appendix~\ref{app:littlestone-alg}).

\section{Preliminaries} \label{sec:prelim}
% !TeX root = main.tex

For a universe (or ground set) $\cU$, a concept $C$ is simply a subset of $\cU$ and a concept class $\cC$ is a collection of concepts. For convenience, we sometimes relax the definition and allow the concepts to not be subsets of $\cU$; all definitions here extend naturally to this case.

The VC and Littlestone's Dimensions can be defined as follows.

\begin{definition}[VC Dimension~\cite{VC71}]
A subset $S \subseteq \cU$ is said to be \emph{shattered} by a concept class $\cC$ if, for every $T \subseteq S$, there exists a concept $C \in \cC$ such that $T = S \cap C$.

The VC Dimension $\VC(\cC, \cU)$ of a concept class $\cC$ with respect to the universe $\cU$ is the largest $d$ such that there exists a subset $S \subseteq \cU$ of size $d$ that is shattered by $\cC$.
\end{definition}

\begin{definition}[Mistake Tree and Littlestone's Dimension~\cite{Lit88}]
A depth-$d$ instance-labeled tree of $\cU$ is a full binary tree of depth $d$ such that every internal node of the tree is assigned an element of $\cU$. For convenience, we will identify each node in the tree canonically by a binary string $s$ of length at most $d$.

A depth-$d$ mistake tree (aka shattered tree~\cite{BDPS09}) for a universe $\cU$ and a concept class $\cC$ is a depth-$d$ instance-labeled tree of $\cU$ such that, if we let $v_s \in \cU$ denote the element assigned to the vertex $s$ for every $s \in \{0, 1\}^{< d}$, then, for every leaf $\ell \in \{0, 1\}^d$, there exists a concept $C \in \cC$ that agrees with the path from root to it, i.e., that, for every $i < d$, $v_{\ell_{\leqs i}} \in C$ iff $\ell_{i + 1} = 1$ where $\ell_{\leqs i}$ denote the prefix of $\ell$ of length $i$.

The Littlestone's Dimension $\LS(\cC, \cU)$ of a concept class $\cC$ with respect to the universe $\cU$ is defined as the maximum $d$ such that there exists a depth-$d$ mistake tree for $\cU, \cC$.
\end{definition}

An equivalent formulation of Littlestone's Dimension is through mistakes made in online learning, as stated below. This interpretation will be useful in our proof.

\begin{definition}[Mistake Bound]
An online algorithm $\cA$ is an algorithm that, at time step $i$, is given an element $x_i \in \cU$ and the algorithm outputs a prediction $p_i \in \{0, 1\}$ whether $x$ is in the class. After the prediction, the algorithm is told the correct answer $h_i \in \{0, 1\}$. For a sequence $(x_1, h_1), \dots, (x_n, h_n)$, \emph{prediction mistake} of $\cA$ is defined as the number of incorect predictions, i.e., $\sum_{i \in n} \mathds{1}[p_i \ne h_i]$. The \emph{mistake bound} of $\cA$ for a concept class $\cC$ is defined as the maximum prediction mistake of $\cA$ over all the sequences $(x_1, h_1), \dots, (x_n, h_n)$ which corresponds to a concept $C \in \cC$ (i.e. $h_i = \mathds{1}[x_i \in C]$ for all $i \in [n]$).
\end{definition}

\begin{theorem}[\cite{Lit88}] \label{lem:learn-ls}
For any universe $\cU$ and any concept class $\cC$, $\LS(\cC, \cU)$ is equal to the minimum mistake bound of $\cC, \cU$ over all online algorithms. 
\end{theorem}

The following facts are well-know and follow easily from the above definitions.

\begin{fact} \label{fact:dim-comp}
For any universe $\cU$ and concept class $\cC$, we have $$\VC(\cC, \cU) \leqs \LS(\cC, \cU) \leqs \log |\cC|.$$
\end{fact}

\begin{fact} \label{fact:ls-union}
For any two universes $\cU_1, \cU_2$ and any concept class $\cC$, $$\LS(\cC, \cU_1 \cup \cU_2) \leqs \LS(\cC, \cU_1) + \LS(\cC, \cU_2).$$
\end{fact}

\subsection{Label Cover and PCP}

As is standard in hardness of approximation, the starting point for our reductions will be the following problem called Label Cover. 

\begin{definition}[Label Cover]
A Label Cover instance $\cL = (A, B, E, \Sigma, \{\pi_e\}_{e \in E})$ consists of a bipartite graph $(A, B, E)$, an alphabet $\Sigma$, and, for every edge $(a, b) \in E$, a projection constraint $\pi_{(a, b)}: \Sigma \to \Sigma$.

An assignment (aka labeling) for $\cL$ is a function $\phi: A \cup B \rightarrow \Sigma$. The value of $\phi$, $\val_\cL(\phi)$ is defined as the fraction of edges $(a, b) \in E$ such that $\pi_{(a, b)}(\phi(a)) = \phi(b)$; these edges are called \emph{satisfied} edges. The value of the instance $\cL$, $\val(\cL)$, is defined as the maximum value among all assignments $\phi: A \cup B \rightarrow \Sigma$.
\end{definition}

Throughout the paper, we often encounter an assignment that only labels a subset of $A \cup B$ but leaves the rest unlabeled. We refer to such assignment as a \emph{partial assignment} to an instance; more specifically, for any $V \subseteq A \cup B$, a $V$-partial assignment (or partial assignment on $V$) is a function $\phi: V \to \Sigma$. For notational convenience, we sometimes write $\Sigma^V$ to denote the set of all functions from $V$ to $\Sigma$.

%more specifically, a partial assignment $\phi$ is a function from $\phi: A \cup B \rightarrow \Sigma \cup \{\bot\}$ where $\bot$ denote unlabeled. For a subset $V \subseteq A \cup B$, a partial assignment $\phi$ is said to be $V$-partial assignment if $V = \phi^{-1}(\bot)$.

We will use the following version of the PCP Theorem by Moshkovitz and Raz, which reduces 3SAT to the gap version of Label Cover while preserves the size to be almost linear.

\begin{theorem}[Moshkovitz-Raz PCP~\cite{MR10}] \label{thm:mr-pcp}
For every $n$ and every $\nu = \nu(n) > 0$, solving 3SAT on $n$ variables can be reduced to distinguishing between the case that a bi-regular instance of Label Cover with $|A|, |B|, |E| = n^{1 + o(1)} \poly(1/\nu)$ and $|\Sigma| = 2^{\poly(1/\nu)}$ is satisfiable and the case that its value is at most $\nu$.
\end{theorem}

\subsection{Useful Lemmata}

We end this section by listing a couple of lemmata that will be useful in our proofs.
\begin{lemma}[Chernoff Bound]
Let $X_1, \dots, X_n$ be i.i.d. random variables taking value from $\{0, 1\}$ and let $p$ be the probability that $X_i = 1$, then, for any $\delta > 0$, we have
\begin{align*}
\Pr\left[\sum_{i=1}^n X_i \geqs (1 + \delta) np\right] \leqs
\begin{cases}
2^{-\delta^2np/3} & \text{if } \delta < 1, \\
2^{-\delta np/3} & \text{otherwise.}
\end{cases}
\end{align*}
\end{lemma}

\begin{lemma}[Partitioning Lemma~{\cite[Lemma~2.5]{R16-communities}}] \label{lem:partition}
For any bi-regular bipartite graph $G = (A, B, E)$, let $n = |A| + |B|$ and $r = \sqrt{n}/\log n$. When $n$ is sufficiently large, there exists a partition of $A \cup B$ into $U_1, \dots, U_r$ such that
\begin{align*}
\forall i \in [r], \frac{n}{2r} \leqs |U_i| \leqs \frac{2n}{r}
\end{align*}
and
\begin{align*}
\forall i, j \in [r], \frac{|E|}{2r^2} \leqs |(U_i \times U_j) \cap E|, |(U_j \times U_i) \cap E| \leqs \frac{2|E|}{r^2}.
\end{align*}
Moreover, such partition can be found in randomized linear time (alternatively, deterministic $n^{O(\log n)}$ time).
\end{lemma}

\section{Inapproximability of VC Dimension}
% !TeX root = main.tex

In this section, we present our reduction from Label Cover to VC Dimension, stated more formally below. We note that this reduction, together with Moshkovitz-Raz PCP (Theorem~\ref{thm:mr-pcp}), with parameter $\delta =  1/\log n$ gives a reduction from 3SAT on $n$ variables to VC Dimension of size $2^{n^{1/2 + o(1)}}$ with gap $1/2 + o(1)$, which immediately implies Theorem~\ref{thm:VC}.

\begin{theorem} \label{thm:vc-red}
For every $\delta > 0$, there exists a randomized reduction from a bi-regular Label Cover instance $\cL = (A, B, E, \Sigma, \{\pi_e\}_{e \in E})$ such that $|\Sigma| = O_\delta(1)$ to a ground set $\cU$ and a concept class $\cC$ such that, if $n \triangleq |A| + |B|$ and $r \triangleq \sqrt{n} / \log n$, then the following conditions hold for every sufficiently large $n$.
\begin{itemize}
\item (Size) The reduction runs in time $|\Sigma|^{O(|E|\poly(1/\delta)/r)}$ and $|\cC|, |\cU| \leqs |\Sigma|^{O(|E|\poly(1/\delta)/r)}$.
\item (Completeness) If $\cL$ is satisfiable, then $\VC(\cC, \cU) \geqs 2r$.
\item (Soundness) If $\val(\cL) \leqs \delta^2 / 100$, then $\VC(\cC, \cU) \leqs (1 + \delta)r$ with high probability.
\end{itemize}
In fact, the above properties hold with high probability even when $\delta$ and $|\Sigma|$ are not constants, as long as $\delta \geqs \log(1000n\log |\Sigma|) / r$.
\end{theorem}

We remark here that when $\delta = 1 / \log n$, Moshkovitz-Raz PCP produces a Label Cover instance with $|A| = n^{1 + o(1)}, |B| = n^{1 + o(1)}$ and $|\Sigma| = 2^{\polylog(n)}$. For such parameters, the condition $\delta \geqs \log(1000n\log |\Sigma|) / r$ holds for every sufficiently large $n$.

\subsection{A Candidate Reduction (and Why It Fails)}\label{sub:VC-candidate}

To best understand the intuition behind our reduction, we first describe a simpler candidate reduction and explain why it fails, which will lead us to the eventual construction. In this candidate reduction, we start by evoking Lemma~\ref{lem:partition} to partition the vertices $A \cup B$ of the Label Cover instance $\cL = (A, B, E, \Sigma, \{\pi_e\}_{e \in E})$ into $U_1, \dots, U_r$ where $r = \sqrt{n}/\log n$. We then create the universe $\cU$ and the concept class $\cC$ as follows:
\begin{itemize}
\item We make each element in $\cU$ correspond to a partial assignment to $U_i$ for some $i \in [r]$, i.e., we let $\cU = \{x_{i, \sigma_i} \mid i \in [r], \sigma_i \in \Sigma^{U_i}\}$. In the completeness case, we expect to shatter the set of size $r$ that corresponds to a satisfying assignment $\sigma^* \in \Sigma^{A \cup B}$ of the Label Cover instance $\cL$, i.e., $\{x_{i, \sigma^*|_{U_i}} \mid i \in [r]\}$. As for the soundness, our hope is that, if a large set $S \subseteq \cU$ gets shattered, then we will be able to decode an assignment for $\cL$ that satisfies many constraints, which contradicts with our assumption that $\val(\cL)$ is small. Note that the number of elements of $\cU$ in this candidate reduction is at most $r \cdot |\Sigma|^{O(|E|\poly(1/\delta)r)} = 2^{\tO(\sqrt{n})}$ as desired.
\item As stated above, the intended solution for the completeness case is $\{x_{i, \sigma^*|_{U_i}} \mid i \in [r]\}$, meaning that we must have at least one concept corresponding to each subset $I \subseteq [r]$. We will try to make our concepts ``test'' the assignment; for each $I \subseteq [r]$, we will choose a set $T_I \subseteq A \cup B$ of $\tO(\sqrt{n})$ vertices and ``test'' all the constraints within $T_I$. Before we specify how $T_I$ is picked, let us elaborate what ``test'' means: for each $T_I$-partial assignment $\phi_I$ that does not violate any constraints within $T_I$, we create a concept $C_{I, \phi_I}$. This concept contains $x_{i, \sigma_i}$ if and only if $i \in I$ and $\sigma_i$ agrees with $\phi_I$ (i.e. $\phi_I|_{T_I \cap U_i} = \sigma_i|_{T_I \cap U_i}$). Recall that, if a set $S \subseteq \cU$ is shattered, then each $\tS \subseteq S$ is an intersection between $S$ and $C_{I, \phi_I}$ for some $I, \phi_I$. We hope that the $I$'s are different for different $\tS$ so that many different tests have been performed on $S$.

Finally, let us specify how we pick $T_I$. Assume without loss of generality that $r$ is even. We randomly pick a perfect matching between $r$, i.e., we pick a random permutation $\pi_I: [r] \to [r]$ and let $\Big(\pi_I(1), \pi_I(2)\Big), \dots, \Big(\pi_I(r - 1), \pi_I(r)\Big)$ be the chosen matching. We pick $T_I$ such that all the constraints in the matchings, i.e., constraints between $U_{\pi_I(2i - 1)}$ and $U_{\pi_I(2i)}$ for every $i \in [r/2]$, are included. More specifically, for every $i \in [r]$, we include each vertex $v \in U_{\pi_I(2i - 1)}$ if at least one of its neighbors lie in $U_{\pi_I(2i)}$ and we include each vertex $u \in U_{\pi_I(2i)}$ if at least one of its neighbors lie in $U_{\pi_I(2i - 1)}$. By Lemma~\ref{lem:partition}, for every pair in the matching the size of the intersection is at most $\frac{2|E|}{r^2}$, so each concept contains assignments to at most $\frac{2|E|}{r}$ variables; so the total size of the concept class is at most $2^r \cdot |\Sigma|^{\frac{2|E|}{r}}$.
\end{itemize}

Even though the above reduction has the desired size and completeness, it unfortunately fails in the soundness. Let us now sketch a counterexample. For simplicity, let us assume that each vertex in $T_{[r]}$ has a unique neighbor in $T_{[r]}$. Note that, since $T_{[r]}$ has quite small size (only $\tO(\sqrt{n})$), almost all the vertices in $T_{[r]}$ satisfy this property w.h.p., but assuming that all of them satisfy this property makes our life easier.

Pick an assignment $\tsigma \in \Sigma^V$ such that none of the constraints in $T_{[r]}$ is violated. From our unique neighbor assumption, there is always such an assignment. Now, we claim that the set $S_{\tsigma} \triangleq \{x_{i, \tsigma|_{U_i}} \mid i \in [r]\}$ gets shattered. This is because, for every subset $I \subseteq [r]$, we can pick another assignment $\sigma'$ such that $\sigma'$ does not violate any constraint in $T_{[r]}$ and $\sigma'|_{U_i} = \tsigma|_{U_i}$ if and only if $i \in I$. This implies that $\{x_{i, \tsigma|_{U_i}} \mid i \in I\} = S \cap C_{[r], \sigma'}$ as desired. Note here that such $\sigma'$ exists because, for every $i \notin I$, if there is a constraint from a vertex $a \in U_i \cap A$ to another vertex $b \in T_{[r]} \cap B$, then we can change the assignment to $a$ in such a way that the constraint is not violated\footnote{Here we assume that $|\pi_{(a, b)}^{-1}(\tsigma(b))| > 1$; note that this always holds for Label Cover instances produced by Moshkovitz-Raz construction.}; by doing this for every $i \notin I$, we have created the desired $\sigma'$. As a result, $\VC(\cC, \cU)$ can still be as large as $r$ even when the value of $\cL$ is small.

%Consider $\pi_{[r]}$ and pick an assignment $\tsigma$ so that all the constraints in $T_{[r]}$ is satisfied. Note that, since $T_{[r]}$ is of size quite small (only $\tO(\sqrt{n})$), our assumption that $\val(\cL)$ is a small constant does not tell us anything about $T_{[r]}$. We claim that $\{\}$

\subsection{The Final Reduction}

In this subsection, we will describe the actual reduction. To do so, let us first take a closer look at the issue with the above candidate reduction. In the candidate reduction, we can view each $I \subseteq [r]$ as being a seed used to pick a matching. Our hope was that many seeds participate in shattering some set $S$, and that this means that $S$ corresponds to an assignment of high value. However, the counterexample showed that in fact only one seed ($I = [r]$) is enough to shatter a set. To circumvent this issue, we will not use the subset $I$ as our seed anymore. Instead, we create $r$ new elements $y_1, \dots, y_r$, which we will call \emph{test selection elements} to act as seeds; namely, each subset $H \subseteq \cY$ will now be a seed. The benefit of this is that, if $S \subseteq \cY$ is shattered and contains test selection elements $y_{i_1}, \dots, y_{i_t}$, then at least $2^t$ seeds must participate in the shattering of $S$. This is because, for each $H \subseteq \cY$, the intersection of $S$ with any concept corresponding to $H$, when restricted to $\cY$, is always $H \cap \{y_{i_1}, \dots, y_{i_t}\}$. Hence, each subset of $\{y_{i_1}, \dots, y_{i_t}\}$ must come a from different seed.

The only other change from the candidate reduction is that each $H$ will test multiple matchings rather than one matching. This is due to a technical reason: we need the number of matchings, $\ell$, to be large in order get the approximation ratio down to $1/2 + o(1)$; in our proof, if $\ell = 1$, then we can only achieve a factor of $1 - \varepsilon$ to some $\varepsilon > 0$. The full details of the reduction are shown in Figure~\ref{fig:red-vc}.

\begin{figure}[h!]
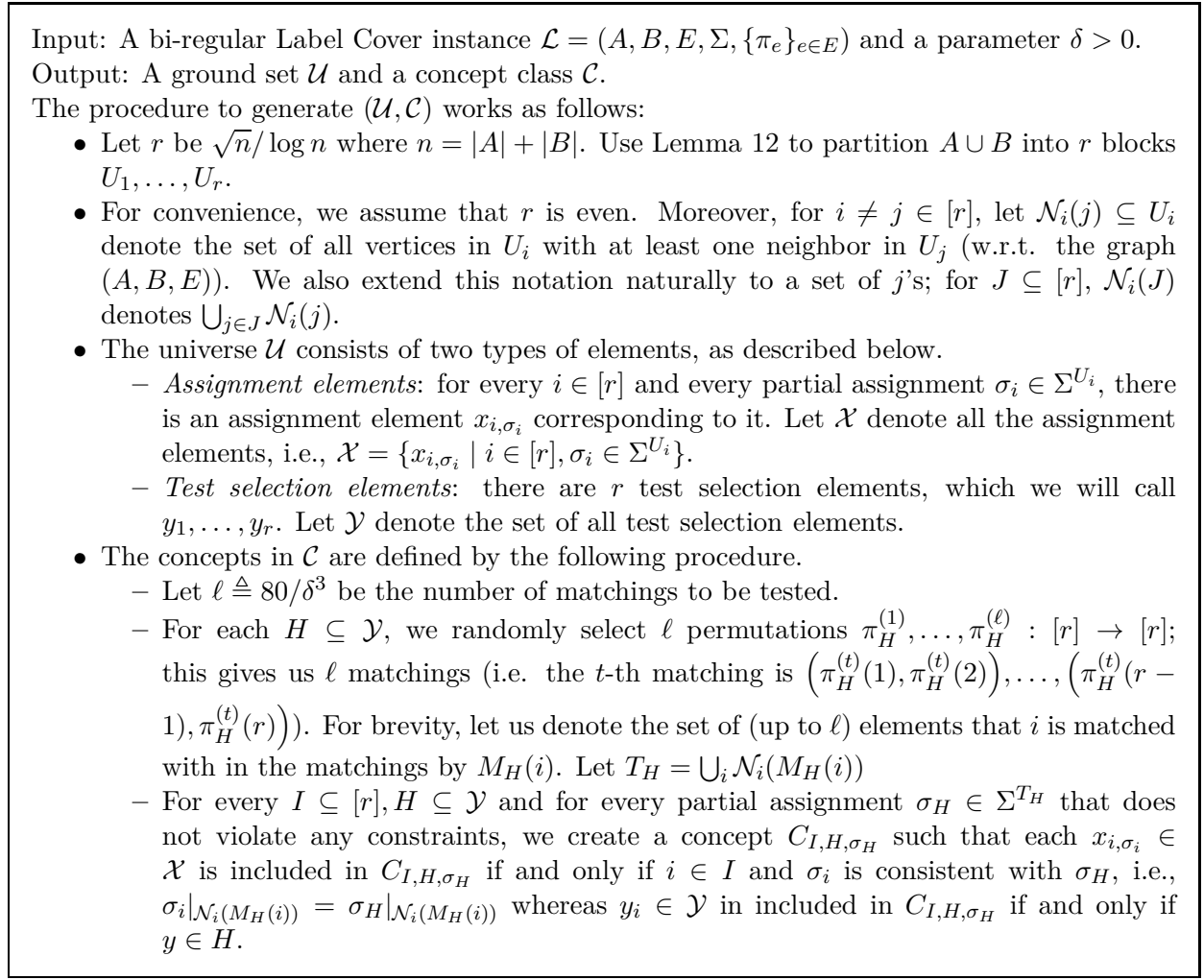

\begin{framed}
Input: A bi-regular Label Cover instance $\cL = (A, B, E, \Sigma, \{\pi_e\}_{e \in E})$ and a parameter $\delta > 0$. \\
Output: A ground set $\cU$ and a concept class $\cC$. \\
The procedure to generate $(\cU, \cC)$ works as follows:
\begin{itemize}
\item Let $r$ be $\sqrt{n} / \log n$ where $n = |A| + |B|$. Use Lemma~\ref{lem:partition} to partition $A \cup B$ into $r$ blocks $U_1, \dots, U_r$.
\item For convenience, we assume that $r$ is even. Moreover, for $i \ne j \in [r]$, let $\cN_i(j) \subseteq U_i$ denote the set of all vertices in $U_i$ with at least one neighbor in $U_j$ (w.r.t. the graph $(A, B, E)$). We also extend this notation naturally to a set of $j$'s; for $J \subseteq [r]$, $\cN_i(J)$ denotes $\bigcup_{j \in J} \cN_i(j)$.
\item The universe $\cU$ consists of two types of elements, as described below.
\begin{itemize}
\item \emph{Assignment elements}: for every $i \in [r]$ and every partial assignment $\sigma_i \in \Sigma^{U_i}$, there is an assignment element $x_{i, \sigma_i}$ corresponding to it. Let $\cX$ denote all the assignment elements, i.e., $\cX = \{x_{i, \sigma_i} \mid i \in [r], \sigma_i \in \Sigma^{U_i}\}$.
\item \emph{Test selection elements}: there are $r$ test selection elements, which we will call $y_1, \dots, y_r$. Let $\cY$ denote the set of all test selection elements.
\end{itemize}
\item The concepts in $\cC$ are defined by the following procedure.
\begin{itemize}
\item Let $\ell \triangleq 80 / \delta^3$ be the number of matchings to be tested.
\item For each $H \subseteq \cY$, we randomly select $\ell$ permutations $\pi^{(1)}_H, \dots, \pi^{(\ell)}_H: [r] \rightarrow [r]$; this gives us $\ell$ matchings (i.e. the $t$-th matching is $\Big(\pi^{(t)}_H(1), \pi^{(t)}_H(2)\Big), \dots, \Big(\pi^{(t)}_H(r - 1), \pi^{(t)}_H(r)\Big)$). For brevity, let us denote the set of (up to $\ell$) elements that $i$ is matched with in the matchings by $M_H(i)$. Let $T_H = \bigcup_i \cN_i(M_{H}(i))$
\item For every $I \subseteq [r], H \subseteq \cY$ and for every partial assignment $\sigma_H \in \Sigma^{T_H}$ that does not violate any constraints, we create a concept $C_{I, H, \sigma_H}$ such that each $x_{i, \sigma_i} \in \cX$ is included in $C_{I, H, \sigma_H}$ if and only if $i \in I$ and $\sigma_i$ is consistent with $\sigma_H$, i.e., $\sigma_i|_{\cN_i(M_{H}(i))} = \sigma_H|_{\cN_i(M_{H}(i))}$ whereas $y_i \in \cY$ in included in $C_{I, H, \sigma_H}$ if and only if $y \in H$.
\end{itemize}
\end{itemize}
\end{framed}
\caption{Reduction from Label Cover to VC Dimension}
\label{fig:red-vc}
\end{figure}

Before we proceed to the proof, let us define some additional notation that will be used throughout.
\begin{itemize}
\item Every assignment element of the form $x_{i, \sigma_i}$ is called an \emph{$i$-assignment element}; we denote the set of all $i$-assignment elements by $\cX_i$, i.e., $\cX_i = \{x_{i, \sigma_i} \mid \sigma_i \in \Sigma^{U_i}\}$.
Let $\cX$ denote all the assignment elements, i.e., $\cX = \bigcup_i \cX_i$.
\item For every $S \subseteq \cU$, let $I(S)$ denote the set of all $i \in [r]$ such that $S$ contains an $i$-assignment element, i.e., $I(S) = \{i \in [r] \mid S \cap \cX_i \ne \emptyset\}$.
\item We call a set $S \subseteq \cX$ \emph{non-repetitive} if, for each $i \in [r]$, $S$ contains at most one $i$-assignment element, i.e., $|S \cap \cX_i| \leqs 1$. Each non-repetitive set $S$ canonically induces a partial assignment $\phi(S): \bigcup_{i \in I(S)} U_i \to \Sigma$. This is the unique partial assignment that satisfies $\phi(S)|_{U_i} = \sigma_i$ for every $x_{i, \sigma_i} \in S$
\item Even though we define each concept as $C_{I, H, \sigma_H}$ where $\sigma_H$ is a partial assignment to a subset $T_H \subseteq A \cup B$, %(that depends on $H$), %an assignment to $N_1(M_{H}(1)) \cup \cdots N_r(M_{H}(r))$ in the above construction, 
it will be more convenient to view each concept as $C_{I, H, \sigma}$ where $\sigma \in \Sigma^V$ is the assignment to the entire Label Cover instance. This is just a notational change: the actual definition of the concept does not depend on the assignment outside $T_H$. %This does not change the definition since, when $\sigma$ does not violate any constraints in $N_1(M_{H}(1)) \cup \cdots N_r(M_{H}(r))$, $C_{I, H, \sigma}$ is exactly equal to $C_{I, H, \sigma_H}$.
\item For each $I \subseteq [r]$, let $U_I$ denote $\bigcup_{i \in I} U_i$. For each $\sigma_I \in \Sigma^{U_I}$, we say that $(I, \sigma_I)$ \emph{passes} $H \subseteq \cY$ if $\sigma_I$ does not violate any constraint within $T_H$. Denote the collection of $H$'s that $(I, \sigma_I)$ passes by $\cH(I, \sigma_I)$.
\item Finally, for any non-repetitive set $S \subseteq \cX$ and any $H \subseteq \cY$, we say that $S$ \emph{passes} $H$ if $(I(S), \phi(S))$ passes $H$. We write $\cH(S)$ as a shorthand for $\cH(I(S), \phi(S))$.
\end{itemize}

The output size of the reduction and the completeness follow almost immediately from definition.

{\bf Output Size of the Reduction.} Clearly, the size of $\cU$ is $\sum_{i \in [r]} |\Sigma|^{|U_i|} \leqs r \cdot |\Sigma|^{n/r} \leqs |\Sigma|^{O(|E|\poly(1/\delta)/r)}$. As for $|\cC|$, note first that the number of choices for $I$ and $H$ are both $2^r$. For fixed $I$ and $H$, Lemma~\ref{lem:partition} implies that, for each matching $\pi_H^{(t)}$, the number of vertices from each $U_i$ with at least one constraint to the matched partition in $\pi_H^{(t)}$ is at most $O(|E|/r^2)$. Since there are $\ell$ matchings, the number of vertices in $T_H = \cN_1(M_H(1)) \cup \cdots \cup \cN_r(M_H(r))$ is at most $O(|E|\ell/r)$. Hence, the number of choices for the partial assignment $\sigma_H$ is at most $|\Sigma|^{O(|E|\poly(1/\delta)/r)}$. In total, we can conclude that $\cC$ contains at most $|\Sigma|^{O(|E|\poly(1/\delta)/r)}$ concepts. %Lastly, it is clear that, apart from the partitioning step that takes $n^{O(\log n)} = o(|\cC|)$ time due to Lemma~\ref{lem:partition}, other steps take time $O(\poly(|\cC|, |\cU|))$ time.

{\bf Completeness.} If $\cL$ has a satisfying assignment $\sigma^* \in \Sigma^V$, then the set $S_{\sigma^*} = \{x_{i, \sigma^*|_{U_i}} \mid i \in [r]\} \cup \cY$ is shattered because, for any $S \subseteq S_{\sigma^*}$, we have $S = S_{\sigma^*} \cap C_{I(S), S \cap \cY, \sigma^*}$. Hence, $\VC(\cC, \cU) \geqs 2r$.

The rest of this section is devoted to the soundness analysis.

\subsection{Soundness}

In this subsection, we will prove the following lemma, which, combined with the completeness and output size arguments above, imply Theorem~\ref{thm:vc-red}.

\begin{lemma} \label{lem:soundness-vc}
Let $(\cC, \cU)$ be the output from the reduction in Figure~\ref{fig:red-vc} on input $\cL$. If $\val(\cL) \leqs \delta^2 / 100$ and $\delta \geqs \log(1000n\log |\Sigma|) / r$, then $\VC(\cC, \cU) \leqs (1 + \delta)r$ w.h.p.
\end{lemma}

At a high level, the proof of Lemma~\ref{lem:soundness-vc} has two steps:
\begin{enumerate}
\item Given a shattered set $S \subseteq \cU$, we extract a maximal non-repetitive set $\Snr \subseteq S$ such that $\Snr$ passes many ($\geqs 2^{|S| - |\Snr|}$) $H$'s. If $|\Snr|$ is small, the trivial upper bound of $2^r$ on the number of different $H$'s implies that $|S|$ is also small. As a result, we are left to deal with the case that $|\Snr|$ is large.
\item When $|\Snr|$ is large, $\Snr$ induces a partial assignment on a large fraction of vertices of $\cL$. Since we assume that $\val(\cL)$ is small, this partial assignment must violate many constraints.  We will use this fact to argue that, with high probability, $\Snr$ only passes very few $H$'s, which implies that $|S|$ must be small.
\end{enumerate}

The two parts of the proof are presented in Subsection~\ref{subsubsec:soundness-I} and~\ref{subsubsec:soundness-II} respectively. We then combine them in Subsection~\ref{subsubsec:soundness-comb} to prove Lemma~\ref{lem:soundness-vc}.

\subsubsection{Part I: Finding a Non-Repetitive Set That Passes Many Tests} \label{subsubsec:soundness-I}

The goal of this subsection is to prove the following lemma, which allows us to, given a shattered set $S \subseteq \cU$, find a non-repetitive set $\Snr$ that passes many $H$'s.

\begin{lemma} \label{lem:many-tests}
For any shattered $S \subseteq \cU$, there is a non-repetitive set $\Snr$ of size $|I(S)|$ \st $|\cH(\Snr)| \geqs 2^{|S| - |I(S)|}$.
\end{lemma}

We will start by proving the following lemma, which will be a basis for the proof of Lemma~\ref{lem:many-tests}. 

\begin{lemma} \label{lem:soundness-mult}
Let $C, C' \in \cC$ correspond to the same $H$ (i.e. $C = C_{I, H, \sigma}$ and $C' = C_{I', H, \sigma'}$ for some $H \subseteq \cY, I, I' \subseteq [r], \sigma, \sigma' \in \Sigma^V$).

For any subset $S \subseteq \cU$ and any maximal non-repetitive subset $\Snr \subseteq S$, if $\Snr \subseteq C$ and $\Snr \subseteq C'$, then $S \cap C = S \cap C'$.
\end{lemma}

The most intuitive interpretation of this lemma is as follows. Recall that if $S$ is shattered, then, for each $\tS \subseteq S$, there must be a concept $C_{I_{\tS}, H_{\tS}, \sigma_{\tS}}$ such that $\tS = S \cap C_{I_{\tS}, H_{\tS}, \sigma_{\tS}}$. The above lemma implies that, for each $\tS \supseteq \Snr$, $H_{\tS}$ must be different. This means that at least $2^{|S| - |\Snr|}$ different $H$'s must be involved in shattering $S$. Indeed, this will be the argument we use when we prove Lemma~\ref{lem:many-tests}.

\begin{proofof}[Lemma~\ref{lem:soundness-mult}]
Let $S, \Snr$ be as in the lemma statement. Suppose for the sake of contradiction that there exists  $H \subseteq \cY, I, I' \subseteq [r], \sigma, \sigma' \in \Sigma^V$ such that $\Snr \subseteq  C_{I, H, \sigma}, \Snr \subseteq  C_{I', H, \sigma'}$ and $S \cap C_{I, H, \sigma} \ne S \cap C_{I', H, \sigma'}$.

First, note that $S \cap C_{I, H, \sigma} \cap \cY = S \cap H \cap \cY = S \cap C_{I', H, \sigma'} \cap \cY$. Since $S \cap C_{I, H, \sigma} \ne S \cap C_{I', H, \sigma'}$, we must have $S \cap C_{I, H, \sigma} \cap \cX \ne S \cap C_{I', H, \sigma'} \cap \cX$. Assume w.l.o.g. that there exists $x_{i, \sigma_i} \in (S \cap C_{I, H, \sigma}) \setminus ( S \cap C_{I', H, \sigma'})$.

Note that $i \in I(S) = I(\Snr)$ (where the equality follows from maximality of $\Snr$). Thus there exists $\sigma'_i \in \Sigma^{U_i}$ such that $x_{i, \sigma'_i} \in \Snr \subseteq C_{I, H, \sigma} \cap  C_{I', H, \sigma'}$. Since $x_{i, \sigma'_i}$ is in both $C_{I, H, \sigma}$ and $C_{I', H, \sigma'}$, we have $i \in I \cap I'$ and
\begin{align} \label{eq:assign-eq}
\sigma|_{\cN_i(M_H(i))} = \sigma'_i|_{\cN_i(M_H(i))} = \sigma'|_{\cN_i(M_H(i))}.
\end{align}
However, since $x_{i, \sigma_i} \in (S \cap C_{I, H, \sigma}) \setminus ( S \cap C_{I', H, \sigma'})$, we have $x_{i, \sigma_i} \in C_{I, H, \sigma} \setminus C_{I', H, \sigma'}$. This implies that $$\sigma|_{\cN_i(M_H(i))} = \sigma_i|_{\cN_i(M_H(i))} \ne \sigma'|_{\cN_i(M_H(i))},$$
which contradicts to (\ref{eq:assign-eq}).
\end{proofof}

In addition to the above lemma, we will also need the following observation, which states that, if a non-repetitive $\Snr$ is contained in a concept $C_{I, H, \sigma_H}$, then $\Snr$ must pass $H$. This observation follows definitions.%from the definition of the concepts, since the concept $C_{I, H, \sigma_H}$ checks all the constraints in $\bigcup_{i \in [r]} \cN_i(M_H(i))$.

\begin{observation} \label{obs:pass-non-rep}
If a non-repetitive set $\Snr$ is a subset of some concept $C_{I, H, \sigma_H}$, then $H \in \cH(\Snr)$.
\end{observation}

With Lemma~\ref{lem:soundness-mult} and Observation~\ref{obs:pass-non-rep} ready, it is now easy to prove Lemma~\ref{lem:many-tests}.

\begin{proofof}[Lemma~\ref{lem:many-tests}]
Pick $\Snr$ to be any maximal non-repetitive subset of $S$. Clearly, $|\Snr| = |I(S)|$. To see that $|\cH(\Snr)| \geqs 2^{|S| - |I(S)|}$, consider any $\tS$ such that $\Snr \subseteq \tS \subseteq S$. Since $S$ is shattered, there exists $I_{\tS}, H_{\tS}, \sigma_{\tS}$ such that $S \cap C_{I_{\tS} , H_{\tS} , \sigma_{\tS}} = \tS$. Since $\tS \supseteq \Snr$, Observation~\ref{obs:pass-non-rep} implies that $H_{\tS} \in \cH(\Snr)$. Moreover, from Lemma~\ref{lem:soundness-mult}, $H_{\tS}$ is distinct for every $\tS$. As a result, $|\cH(\Snr)| \geqs 2^{|S| - |I(S)|}$ as desired.
\end{proofof}

\subsubsection{Part II: No Large Non-Repetitive Set Passes Many Tests} \label{subsubsec:soundness-II}

The goal of this subsection is to show that, if $\val(\cL)$ is small, then w.h.p. (over the randomness in the construction) every large non-repetitive set passes only few $H$'s. This is formalized as Lemma~\ref{lem:passing-bound} below.

\begin{lemma} \label{lem:passing-bound}
If $\val(\cL) \leqs \delta^2/100$ and $\delta \geqs 8 / r$, then, with high probability, for every non-repetitive set $\Snr$ of size at least $\delta r$, $|\cH(\Snr)| \leqs 100n \log |\Sigma|$.
\end{lemma}

Note that the mapping $\Snr \mapsto (I(\Snr), \phi(\Snr))$ is a bijection from the collection of all non-repetitive sets to $\{(I, \sigma_I) \mid I \subseteq [r], \sigma_I \in \Sigma^{U_I}\}$. Hence, the above lemma is equivalent to the following.

\begin{lemma} \label{lem:passing-bound-equiv-form}
If $\val(\cL) \leqs \delta^2/100$ and $\delta \geqs8 / r$, then, with high probability, for every $I \subseteq [r]$ of size at least $\delta r$ and every $\sigma_I \in \Sigma^{U_I}$, $|\cH(I, \sigma_I)| \leqs 100n \log |\Sigma|$.
\end{lemma}

Here we use the language in Lemma~\ref{lem:passing-bound-equiv-form} instead of Lemma~\ref{lem:passing-bound} as it will be easier for us to reuse this lemma later. To prove the lemma, we first need to bound the probability that each assignment $\sigma_I$ does not violate any constraint induced by a random matching. More precisely, we will prove the following lemma.

\begin{lemma} \label{lem:passing-bound-single}
For any $I \subseteq [r]$ of size at least $\delta r$ and any $\sigma_I \in \Sigma^{U_I}$, if $\pi: [r] \rightarrow [r]$ is a random permutation of $[r]$, then the probability that $\sigma_I$ does not violate any constraint in $\bigcup_{i \in [r]} \cN_i(M(i))$ is at most $(1 - 0.1\delta^2)^{\delta r / 8}$ where $M(i)$ denote the index that $i$ is matched with in the matching $\Big(\pi(1), \pi(2)\Big), \dots, \Big(\pi(r - 1), \pi(r)\Big)$.
\end{lemma}

\begin{proof}
Let $p$ be any positive odd integer such that $p \leqs \delta r / 2$ and let $i_1, \dots, i_{p - 1} \in [r]$ be any $p - 1$ distinct elements of $[r]$. We will first show that conditioned on $\pi(1) = i_1, \dots, \pi(p - 1) = i_{p - 1}$, the probability that $\sigma_I$ violates a constraint induced by $\pi(p), \pi(p + 1)$ (i.e. in $\cN_{\pi(p)}(\pi(p + 1)) \cup \cN_{\pi(p + 1)}(\pi(p))$) is at least $0.1\delta^2$.

To see that this is true, let $I_{\geqs p} = I \setminus \{i_1, \dots, i_{p - 1}\}$. Since $|I| \geqs \delta r$, we have $|I_{\geqs p}| = |I| - p + 1 \geqs \delta r / 2 + 1$. Consider the partial assignment $\sigma_{\geqs p} = \sigma_I|_{U_{I_{\geqs p}}}$. Since $\val(\cL) \leqs 0.01\delta^2$, $\sigma_{\geqs p}$ can satisfy at most $0.01\delta^2 |E|$ constraints. From Lemma~\ref{lem:partition}, we have, for every $i \ne j \in I_{\geqs p}$, the number of constraints between $U_i$ and $U_j$ are at least $|E|/r^2$. Hence, there are at most $0.01\delta^2r^2$ pairs of $i < j \in I_{\geqs p}$ such that $\sigma_{\geqs p}$ does not violate any constraint between $U_i$ and $U_j$. In other words, there are at least $\binom{|I_{\geqs p}|}{2} - 0.01\delta^2r^2 \geqs 0.1\delta^2r^2$ pairs $i < j \in I_{\geqs p}$ such that $\sigma_{\geqs p}$ violates some constraints between $U_i$ and $U_j$. Now, if $\pi(p) = i$ and $\pi(p + 1) = j$ for some such pair $i, j$, then $\phi(\Snr)$ violates a constraint induced by $\pi(p), \pi(p + 1)$. Thus, we have
\begin{align} \label{eq:one-match}
\Pr\left[\sigma_I \text{ does not violate a constraint induced by } \pi(p), \pi(p + 1) \midv \bigwedge_{t=1}^{p - 1} \pi(t) = i_t\right] \leqs 1 - 0.1\delta^2.
\end{align}

Let $E_p$ denote the event that $\sigma_I$ does not violate any constraints induced by $\pi(p)$ and $\pi(p + 1)$. We can now bound the desired probability as follows.
\begin{align*}
\Pr\left[\sigma_I \text{ does not violate any constraint in } \bigcup_{i \in [r]}\cN_i(M(i))\right] &\leqs \Pr\left[\bigwedge_{\text{odd } p \in [\delta r / 2 + 1]} E_p\right] \\
&= \prod_{\text{odd } p \in [\delta r / 2 + 1]} \Pr\left[E_p \midv \bigwedge_{\text{odd } t \in [p - 1]} E_t \right] \\
(\text{From } (\ref{eq:one-match})) &\leqs \prod_{\text{odd } p \in [\delta r / 2 + 1]} (1 - 0.1\delta^2) \\
&\leqs (1 - 0.1\delta^2)^{\delta r / 4 - 1},
\end{align*}
which is at most $(1 - 0.1\delta^2)^{\delta r / 8}$ since $\delta \geqs 8/r$.
\end{proof}

We can now prove our main lemma.

\begin{proofof}[Lemma~\ref{lem:passing-bound-equiv-form}]
For a fixed $I \subseteq [r]$ of size at least $\delta r$ and a fixed $\sigma_I \in \Sigma^{U_I}$, Lemma~\ref{lem:passing-bound-single} tells us that the probability that $\sigma_I$ does not violate any constraint induced by a single matching is at most $(1 - 0.1\delta^2)^{\delta r/8}$. Since for each $H \subseteq \cY$ the construction picks $\ell$ matchings at random, the probability that $(I, \sigma_I)$ passes each $H$ is at most $(1 - 0.1\delta^2)^{\delta \ell r/8}$. Recall that we pick $\ell = 80/\delta^3$; this gives the following upper bound on the probability:
\begin{align}\label{eq:Pr[H]}
\Pr[(I, \sigma_I) \text{\;passes\;} H] \leq 
(1 - 0.1\delta^2)^{\delta \ell r/8} = (1 - 0.1\delta^2)^{10r/\delta^2} &\leqs \left(\frac{1}{1 + 0.1\delta^2}\right)^{10r/\delta^2} \leqs 2^{-r}
\end{align}
where the last inequality comes from Bernoulli's inequality.

Inequality~\eqref{eq:Pr[H]} implies that the expected number of $H$'s that $(I, \sigma_I)$ passes is less than $1$.
Since the matchings $M_H$ are independent for all $H$'s, we can apply Chernoff bound which implies that
\begin{align*}
\Pr[|\cH(I, \sigma_I)| \geqs 100n \log |\Sigma|] \leqs 2^{-10n \log |\Sigma|} = |\Sigma|^{-10 n}.
\end{align*}

Finally, note that there are at most $2^r |\Sigma|^n$ different $(I, \sigma_I)$'s. By union bound, we have
\begin{align*}
\Pr\Big[\exists I \subseteq [r], \sigma_I \in \Sigma^{U_I} \;\; \text{s.t.}\;\; |I| \geqs \delta r \;\;\text{AND}\;\; |\cH(I, \sigma_I)| \geqs 100n \log |\Sigma| \Big] 
& \leqs \left(2^r |\Sigma|^n\right)\left(|\Sigma|^{-10n}\right) \\
& \leqs |\Sigma|^{-8n},
\end{align*}
which concludes the proof.
\end{proofof}

\subsubsection{Putting Things Together} \label{subsubsec:soundness-comb}

\begin{proofof}[Lemma~\ref{lem:soundness-vc}]
From Lemma~\ref{lem:passing-bound}, every non-repetitive set $\Snr$ of size at least $\delta r$, $|\cH(\Snr)| \leqs 100 n \log |\Sigma|$. Conditioned on this event happening, we will show that $\VC(\cU, \cC) \leqs (1 + \delta) r$. 

Consider any shattered set $S \subseteq \cU$. Lemma~\ref{lem:many-tests} implies that there is a non-repetitive set $\Snr$ of size $|I(S)|$ such that $|\cH(\Snr)| \geqs 2^{|S| - |I(S)|}$. Let us consider two cases:
\begin{enumerate}
\item $|I(S)| \leqs \delta r$. Since $\cH(\Snr) \subseteq \cP(\cY)$, we have $|S| - |I(S)| \leqs |\cY| = r$. This implies that $|S| \leqs (1 + \delta)r$.
\item $|I(S)| > \delta r$. From our assumption, $|\cH(\Snr)| \leqs 100 n \log |\Sigma|$. Thus, $|S| \leqs |I(S)| + \log(100 n \log |\Sigma|) \leqs (1 + \delta)r$ where the second inequality comes from our assumption that $\delta \geqs \log(1000n\log |\Sigma|) / r$.
\end{enumerate}
Hence, $\VC(\cU, \cC) \leqs (1 + \delta)r$ with high probability.
\end{proofof}

\section{Inapproximability of Littlestone's Dimension} \label{sec:littlestone}
% !TeX root = main.tex

We next proceed to Littlestone's Dimension. The main theorem of this section is stated below. Again, note that this theorem and Theorem~\ref{thm:mr-pcp} implies Theorem~\ref{thm:LS}.

\begin{theorem} \label{thm:ls-red}
There exists $\varepsilon > 0$ such that there is a randomized reduction from any bi-regular Label Cover instance $\cL = (A, B, E, \Sigma, \{\pi_e\}_{e \in E})$ with $|\Sigma| = O(1)$ to a ground set $\cU$ and a concept classes $\cC$ such that, if $n \triangleq |A| + |B|, r \triangleq \sqrt{n} / \log n$ and $k \triangleq 10^{10}|E|\log |\Sigma| / r^2$, then the following conditions hold for every sufficiently large $n$.
\begin{itemize}
\item (Size) The reduction runs in time $2^{rk} \cdot |\Sigma|^{O(|E|/r)}$ and $|\cC|, |\cU| \leqs 2^{rk} \cdot |\Sigma|^{O(|E|/r)}$.
\item (Completeness) If $\cL$ is satisfiable, then $\LS(\cC, \cU) \geqs 2rk$.
\item (Soundness) If $\val(\cL) \leqs 0.001$, then $\LS(\cC, \cU) \leqs (2 - \varepsilon)rk$ with high probability.
\end{itemize}
\end{theorem}

\subsection{Why the VC Dimension Reduction Fails for Littlestone's Dimension}

It is tempting to think that, since our reduction from the previous section works for VC Dimension, it may also work for Littlestone's Dimension. In fact, thanks to Fact~\ref{fact:dim-comp}, completeness for that reduction even translates for free to Littlestone's Dimension. Alas, the soundness property does not hold. To see this, let us build a depth-$2r$ mistake tree for $\cC, \cU$, even when $\val(\cL)$ is small, as follows.
\begin{itemize}
\item We assign the test-selection elements to the first $r$ levels of the tree, one element per level. More specifically, for each $s \in \{0, 1\}^{< r}$, we assign $y_{|s| + 1}$ to $s$.
\item For every string $s \in \{0, 1\}^r$, the previous step of the construction gives us a subset of $\cY$ corresponding to the path from root to $s$; this subset is simply $H_s = \{y_i \in \cY \mid s_i = 1\}$. Let $T_{H_s}$ denote the set of vertices tested by this seed $H_s$. Let $\phi_s \in \Sigma^V$ denote an assignment that satisfies all the constraints in $T_{H_s}$. Note that, since $T_{H_s}$ is of small size (only $\tO(\sqrt{n})$), even if $\val(\cL)$ is small, $\phi_s$ is still likely to exist (and we can decide whether it exists or not in time $2^{\tilde O(\sqrt{n})}$).

We then construct the subtree rooted at $s$ that corresponds to $\phi_s$ by assigning each level of the subtree $x_{i, \phi_s|_{U_i}}$. Specifically, for each $t \in \{0, 1\}^{\geqs r}$, we assign $x_{|t| - r + 1, \phi_{t_{\leqs r}}|_{U_{|t| - r + 1}}}$ to node $t$ of the tree.
\end{itemize}

It is not hard to see that the constructed tree is indeed a valid mistake tree. This is because the path from root to each leaf $l \in \{0, 1\}^{2r}$ agrees with $C_{I(l), H_{l_{\leqs r}}, \phi_{l \leqs r}}$ (where $I(l) = \{i \in [r] \mid l_i = 1\}$).

\subsection{The Final Reduction}

The above counterexample demonstrates the main difference between the two dimensions: order does not matter in VC Dimension, but it does in Littlestone's Dimension. By moving the test-selection elements up the tree, the tests are chosen before the assignments, which allows an adversary to ``cheat'' by picking different assignments for different tests. We would like to prevent this, i.e., we would like to make sure that, in the mistake tree, the upper levels of the tree are occupied with the assignment elements whereas the lower levels are assigned test-selection elements. As in the VC Dimension argument, our hope here is that, given such a tree, we should be able to decode an assignment that passes tests on many different tests. Indeed we will tailor our construction to achieve such property.

%Indeed, we will use ordering in the tree to our advantage in our reduction.
Recall that, if we use the same reduction as VC Dimension, then, in the completeness case, we can construct a mistake tree in which the first $r$ layers consist solely of assignment elements and the rest of the layers consist of only test-selection elements. Observe that there is no need for different nodes on the $r$-th layer to have subtrees composed of the same set of elements; the tree would still be valid if we make each test-selection element only work with a specific $s \in \{0, 1\}^r$ and create concepts accordingly. In other words, we can modify our construction so that our test-selection elements are $\cY = \{y_{I, i} \mid I \subseteq [r], i \in [r]\}$ and the concept class is $\{C_{I, H, \sigma_H} \mid I \subseteq [r], H \subseteq \cY, \sigma_H \in \Sigma^{T_H}\}$ where the condition that an assignment element lies in $C_{I, H, \sigma_H}$ is the same as in the VC Dimension reduction, whereas for $y_{I', i}$ to be in $C_{I, H, \sigma_H}$, we require not only that $i \in H$ but also that $I = I'$. Intuitively, this should help us, since each $y_{I, i}$ is now only in a small fraction ($\leqs 2^{-r}$) of concepts; hence, one would hope that any subtree rooted at any $y_{I, i}$ cannot be too deep, which would indeed implies that the test-selection elements cannot appear in the first few layers of the tree.

Alas, for this modified reduction, it is not true that a subtree rooted at any $y_{I, i}$ has small depth; specifically, we can bound the depth of a subtree $y_{I, i}$ by the log of the number of concepts containing $y_{I, i}$ plus one (for the first layer). Now, note that $y_{I, i} \in C_{I', H, \sigma_H}$ means that $I' = I$ and $i \in H$, but there can be still as many as $2^{r - 1} \cdot |\Sigma|^{|T_H|} = |\Sigma|^{O(|E|/r)}$ such concepts. This gives an upper bound of $r + O(|E|\log|\Sigma|/r)$ on the depth of the subtree rooted at $y_{I, i}$. However, $|E|\log|\Sigma|/r = \Theta(\sqrt{n}\log n) = \omega(r)$; this bound is meaningless here since, even in the completeness case, the depth of the mistake tree is only $2r$.

Fortunately, this bound is not useless after all: if we can keep this bound but make the intended tree depth much larger than $|E|\log|\Sigma|/r$, then the bound will indeed imply that no $y_{I, i}$-rooted tree is deep. To this end, our reduction will have one more parameter $k = \Theta(|E|\log|\Sigma|/r)$ where $\Theta(\cdot)$ hides a large constant and the intended tree will have depth $2rk$ in the completeness case; the top half of the tree (first $rk$ layers) will again consist of assignment elements and the rest of the tree composes of the test-selection elements. The rough idea is to make $k$ ``copies'' of each element: the assignment elements will now be $\{x_{i, \sigma_i, j} \mid i \in [r], \sigma_i \in \Sigma^{U_i}, j \in [k]\}$ and the test-selection elements will be $\{y_{I, i, j} \mid I \subseteq [r] \times [k], j \in [k]\}$. The concept class can then be defined as $\{C_{I, H, \sigma_H} \mid I \subseteq [r] \times [k], H \subseteq [r] \times [k], \sigma_H \in \Sigma^{T_H}\}$ naturally, i.e., $H$ is used as the seed to pick the test set $T_H$, $y_{I', i, j} \in C_{I, H, \sigma_H}$ iff $I' = I$ and $(i, j) \in H$ whereas $x_{i, \sigma_i, j} \in C_{I, H, \sigma_H}$ iff $(i, j) \in I$ and $\sigma_i|_{(I, \sigma_I)} = \sigma_H|_{(I, \sigma_I)}$. For this concept class, we can again bound the depth of $y_{I, i}$-rooted tree to be $rk + O(|E|\log|\Sigma|/r)$; this time, however, $rk$ is much larger than $|E|\log|\Sigma|/r$, so this bound is no more than, say, $1.001rk$. This is indeed the desired bound, since this means that, for any depth-$1.999rk$ mistake tree, the first $0.998rk$ layers must consist solely of assignment elements.

Unfortunately, the introduction of copies in turn introduces another technical challenge: it is not true any more that a partial assignment to a large set only passes a few tests w.h.p. (i.e. an analogue of Lemma~\ref{lem:passing-bound-equiv-form} does not hold). By Inequality~\eqref{eq:Pr[H]}, each $H$ is passed with probability at most $2^{-r}$, but now we want to take a union bound there are $2^{rk} \gg 2^r$ different $H$'s. % This is due to the fact that the seeds are now $rk$ bits instead of $r$ bits. 
To circumvent this, we will define a map $\tau: \cP([r] \times [k]) \to \cP([r])$ and use $\tau(H)$ to select the test instead of $H$ itself. The map $\tau$ we use in the construction is the \emph{threshold projection} where $i$ is included in $H$ if and only if, for at least half of $j \in [k]$, $H$ contains $(i, j)$. To motivate our choice of $\tau$, recall that our overall proof approach is to first find a node that corresponds to an assignment to a large subset of the Label Cover instance; then argue that it can pass only a few tests, which we hope would imply that the subtree rooted there cannot be too deep. For this implication to be true, we need the following to also hold: for any small subset $\cH \subseteq \cP([r])$ of $\tau(H)$'s, we have that $\LS(\tau^{-1}(\cH), [r] \times [k])$ is small. This property indeed holds for our choice of $\tau$ (see Lemma~\ref{lem:non-rep-shallow-y}).

With all the moving parts explained, we state the full reduction formally in Figure~\ref{fig:red-ls}.

\begin{figure}[h!]
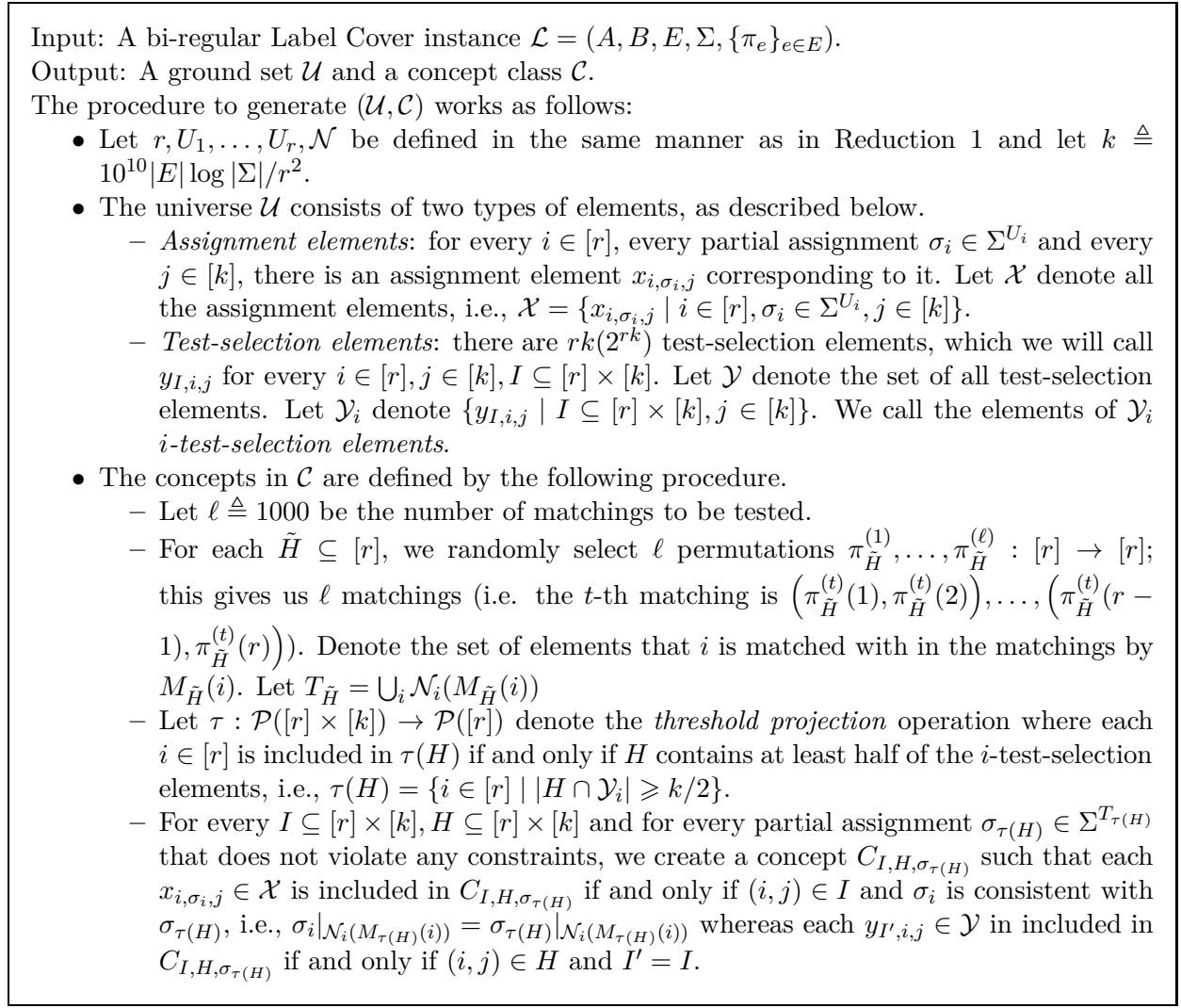

\begin{framed}
Input: A bi-regular Label Cover instance $\cL = (A, B, E, \Sigma, \{\pi_e\}_{e \in E})$. \\
Output: A ground set $\cU$ and a concept class $\cC$. \\
The procedure to generate $(\cU, \cC)$ works as follows:
\begin{itemize}
\item Let $r, U_1, \dots, U_r, \cN$ be defined in the same manner as in Reduction~\ref{fig:red-vc} and let $k \triangleq 10^{10}|E|\log |\Sigma| / r^2$.
\item The universe $\cU$ consists of two types of elements, as described below.
\begin{itemize}
\item \emph{Assignment elements}: for every $i \in [r]$, every partial assignment $\sigma_i \in \Sigma^{U_i}$ and every $j \in [k]$, there is an assignment element $x_{i, \sigma_i, j}$ corresponding to it. Let $\cX$ denote all the assignment elements, i.e., $\cX = \{x_{i, \sigma_i, j} \mid i \in [r], \sigma_i \in \Sigma^{U_i}, j \in [k]\}$.
\item \emph{Test-selection elements}: there are $rk(2^{rk})$ test-selection elements, which we will call $y_{I, i, j}$ for every $i \in [r], j \in [k], I \subseteq [r] \times [k]$. Let $\cY$ denote the set of all test-selection elements. Let $\cY_i$ denote $\{y_{I, i, j} \mid I \subseteq [r] \times [k], j \in [k]\}$. We call the elements of $\cY_i$ \emph{$i$-test-selection elements}.
\end{itemize}
\item The concepts in $\cC$ are defined by the following procedure.
\begin{itemize}
\item Let $\ell \triangleq 1000$ be the number of matchings to be tested.
\item For each $\tH \subseteq [r]$, we randomly select $\ell$ permutations $\pi^{(1)}_{\tH}, \dots, \pi^{(\ell)}_{\tH}: [r] \rightarrow [r]$; this gives us $\ell$ matchings (i.e. the $t$-th matching is $\Big(\pi^{(t)}_{\tH}(1), \pi^{(t)}_{\tH}(2)\Big), \dots, \Big(\pi^{(t)}_{\tH}(r - 1), \pi^{(t)}_{\tH}(r)\Big)$). Denote the set of elements that $i$ is matched with in the matchings by $M_{\tH}(i)$. Let $T_{\tH} = \bigcup_i \cN_i(M_{\tH}(i))$
\item Let $\tau: \cP([r] \times [k]) \rightarrow \cP([r])$ denote the \emph{threshold projection} operation where each $i \in [r]$ is included in $\tau(H)$ if and only if $H$ contains at least half of the $i$-test-selection elements, i.e., $\tau(H) = \{i \in [r]  \mid |H \cap \cY_i| \geqs k / 2\}$.
\item For every $I \subseteq [r] \times [k], H \subseteq [r] \times [k]$ and for every partial assignment $\sigma_{\tau(H)} \in \Sigma^{T_{\tau(H)}}$ that does not violate any constraints, we create a concept $C_{I, H, \sigma_{\tau(H)}}$ such that each $x_{i, \sigma_i, j} \in \cX$ is included in $C_{I, H, \sigma_{\tau(H)}}$ if and only if $(i, j) \in I$ and $\sigma_i$ is consistent with $\sigma_{\tau(H)}$, i.e., $\sigma_i|_{\cN_i(M_{\tau(H)}(i))} = \sigma_{\tau(H)}|_{\cN_i(M_{\tau(H)}(i))}$ whereas each $y_{I', i, j} \in \cY$ in included in $C_{I, H, \sigma_{\tau(H)}}$ if and only if $(i, j) \in H$ and $I' = I$.
\end{itemize}
\end{itemize}
\end{framed}
\caption{Reduction from Label Cover to Littlestone's Dimension}
\label{fig:red-ls}
\end{figure}

Similar to our VC Dimension proof, we will use the following notation:%s in our proof:
\begin{itemize}
\item For every $i \in [r]$, let $\cX_i \triangleq \{x_{i, \sigma_i, j} \mid \sigma_i \in \Sigma^{U_i}, j \in [k]\}$; we refer to these elements as the $i$-assignment elements. Moreover, for every $(i, j) \in [r] \times [k]$, let $\cX_{i, j} \triangleq \{x_{i, \sigma_i, j} \mid \sigma_i \in \Sigma^{U_i}\}$; we refer to these elements as the $(i, j)$-assignment elements.
\item For every $S \subseteq \cU$, let $I(S) = \{i \in [r] \mid S \cap \cX_i \ne \emptyset\}$ and $IJ(S) = \{(i, j) \in [r] \times [k] \mid S \cap \cX_{i, j} \ne \emptyset\}$.
\item A set $S \subseteq \cX$ is \emph{non-repetitive} if $|S \cap \cX_{i, j}| \leqs 1$ for all $(i, j) \in [r] \times [k]$.
\item We say that $S$ \emph{passes} $\tH$ if the following two conditions hold:
\begin{itemize}
\item For every $i \in [r]$ such that $S \cap \cX_{i} \ne \emptyset$, all $i$-assignment elements of $S$ are consistent on $T_{\tH}|_{U_i}$, i.e., for every $(i, \sigma_i, j), (i, \sigma'_i, j') \in S$, we have $\sigma_i|_{U_i} = \sigma'_i|_{U_i}$.
\item The canonically induced assignment on $T_{\tH}$ does not violate any constraint (note that the previous condition implies that such assignment is unique).
\end{itemize}
We use $\cH(S)$ to denote the collection of all seeds $\tH \subseteq [r]$ that $S$ passes.
\end{itemize}

We also use the following notation for mistake trees:
\begin{itemize}
\item For any subset $S \subseteq \cU$ and any function $\rho: S \to \{0, 1\}$, let $\cC[\rho] \triangleq \{C \in \cC \mid \forall a \in S, a \in C \Leftrightarrow \rho(a) = 1\}$ be the collections of all concept that agree with $\rho$ on $S$. We sometimes abuse the notation and write $\cC[S]$ to denote the collection of all the concepts that contain $S$, i.e., $\cC[S] = \{C \in \cC \mid S \subseteq C\}$.
%\item For two functions $\rho_1: S_1 \to \{0, 1\}$ and $\rho_2: S_2 \to \{0, 1\}$ such that $S_1 \cap S_2 = \emptyset$, we write $\rho_1 \circ \rho_2$ to denote the unique function from $S_1 \cup S_2 \to \{0, 1\}$ such that $(\rho_1 \circ \rho_2)|_{S_1} = \rho_1|_{S_1}$ and $(\rho_1 \circ \rho_2)|_{S_2} = \rho_1|_{S_2}$.
%\item We sometimes write $\{a \mapsto b\}$ to denote the function on $\{a\}$ that maps $a$ to $b$.
\item For any binary string $s$, let $\pre(s) \triangleq \{\emptyset, s_{\leqs 1}, \dots, s_{\leqs |s| - 1}\}$ denote the set of all proper prefixes of $s$. 
\item For any depth-$d$ mistake tree $\cT$, let $v_{\cT, s}$ denote the element assigned to the node $s \in \{0, 1\}^{\leqs d}$, and let $P_{\cT, s} \triangleq \{v_{\cT, s'} \mid s' \in \pre(s)\}$ denote the set of all elements appearing from the path from root to $s$ (excluding $s$ itself). Moreover, let $\rho_{\cT, s}: P_{\cT, s} \to \{0, 1\}$ be the function corresponding to the path from root to $s$, i.e., $\rho_{\cT, s}(v_{\cT, s'}) = s_{|s'| + 1}$ for every $s' \in \pre(s)$.
\end{itemize}

{\bf Output Size of the Reduction} The output size of the reduction follows immediately from a similar argument as in the VC Dimension reduction. The only different here is that there are $2^{rk}$ choices for $I$ and $H$, instead of $2^r$ choices as in the previous construction.

{\bf Completeness.} If $\cL$ has a satisfying assignment $\sigma^* \in \Sigma^V$, we can construct a depth-$rk$ mistake tree $\cT$ as follows. For $i \in [r], j \in [k]$, we assign $x_{i, \sigma^*|_{U_i}, j}$ to every node in the $((i - 1)k + j)$-th layer of $\cT$. Note that we have so far assigned every node in the first $rk$ layers. For the rest of the vertices $s$'s, if $s$ lies in layer $rk + (i - 1)k + j$, then we assign $y_{I(\rho^{-1}_{\cT, s}(1)), i, j}$ to it. It is clear that, for a leaf $s \in \{0, 1\}^{rk}$, the concept $C_{I(\rho^{-1}_{\cT, s}(1)), H_{\cT, s}, \sigma^*}$ agrees with the path from root to $s$ where $H_{\cT, s}$ is defined as $\{(i, j) \in [r] \times [k] \mid y_{I(\rho^{-1}_{\cT, s}(1)), i, j} \in \rho^{-1}_{\cT, s}(1)\}$. Hence, $\LS(\cC, \cU) \geqs 2rk$.

\subsection{Soundness}

Next, we will prove the soundness of our reduction, stated more precisely below. For brevity, we will assume throughout this subsection that $r$ is sufficiently large, and leave it out of the lemmas' statements. Note that this lemma, together with completeness and output size properties we argue above, implies Theorem~\ref{thm:ls-red} with $\varepsilon = 0.001$. 

\begin{lemma} \label{lem:soundness-ls}
Let $(\cC, \cU)$ be the output from the reduction in Figure~\ref{fig:red-ls} on input $\cL$. If $\val(\cL) \leqs 0.001$, then $\LS(\cC, \cU) \leqs 1.999rk$ with high probability.
\end{lemma}

Roughly speaking, the overall strategy of our proof of Lemma~\ref{lem:soundness-ls} is as follows:
\begin{enumerate}
\item First, we will argue that any subtree rooted at any test-selection element must be shallow (of depth $\leqs 1.001rk$). This means that, if we have a depth-$1.999rk$ mistake tree, then the first $0.998rk$ levels must be assigned solely assignment elements.
\item We then argue that, in this $0.998rk$-level mistake tree of assignment elements, we can always extract a leaf $s$ such that the path from root to $s$ indicates inclusion of a large non-repetitive set. In other words, the path to $s$ can be decoded into a (partial) assignment for the Label Cover instance $\cL$.
\item Let the leaf from the previous step be $s$ and the non-repetitive set be $\Snr$. Our goal now is to show that the subtree rooted as $s$ must have small depth. We start working towards this by showing that, with high probability, there are few tests that agree with $\Snr$. This is analogous to Part II of the VC Dimension proof.
\item With the previous steps in mind, we only need to argue that, when $|\cH(\Snr)|$ is small, the Littlestone's dimension of all the concepts that contains $\Snr$ (i.e. $\LS(\cC[\Snr], \cU)$) is small. Thanks to Fact~\ref{fact:ls-union}, it is enough for us to bound $\LS(\cC[\Snr], \cX)$ and $\LS(\cC[\Snr], \cY)$ separately. For the former, our technique from the second step also gives us the desired bound; for the latter, we prove that $\LS(\cC[\Snr], \cY)$ is small by designing an algorithm that provides correct predictions on a constant fraction of the elements in $\cY$.
\end{enumerate}

Let us now proceed to the details of the proofs.

\subsubsection{Part I: Subtree of a Test-Selection Assignment is Shallow}

\begin{lemma} \label{lem:ls-assign-not-high}
For any $y_{I, i, j} \in \cY$, $\LS(\cC[\{y_{I, i, j}\}], \cU) \leqs rk + (4|E|\ell/r)\log |\Sigma| \leqs 1.001 rk$.
\end{lemma}

Note that the above lemma implies that, in any mistake tree, the depth of the subtree rooted at any vertex $s$ assigned to some $y_{I, i, j} \in \cY$ is at most $1 + 1.001rk$. This is because every concept that agrees with the path from the root to $s$ must be in $\cC[\{y_{I, i, j}\}]$, which has depth at most $1.001rk$.

\begin{proofof}[Lemma~\ref{lem:ls-assign-not-high}]
Consider any $C_{I', H, \sigma_{\tau(H)}} \in \cC[\{y_{I, i, j}\}], \cU)$. Since $y_{I, i, j} \in C_{I', H, \sigma_{\tau(H)}}$, we have $I = I'$. Moreover, from Lemma~\ref{lem:partition}, we know that $\left|\cN_i\left(M_{\tau(H)}(i)\right)\right| \leqs 4|E|\ell/r^2$, which implies that $|T_{\tau(H)}| \leqs 4|E|\ell/r$. This means that there are only at most $|\Sigma|^{4|E|\ell/r}$ choices of $\sigma_{\tau(H)}$. Combined with the fact that there are only $2^{rk}$ choices of $H$, we have $|\cC[\{y_{I, i, j}\}]| \leqs 2^{rk} \cdot |\Sigma|^{4|E|\ell/r}$. Fact~\ref{fact:dim-comp} then implies the lemma.
\end{proofof}

\subsubsection{Part II: Deep Mistake Tree Contains a Large Non-Repetitive Set}

The goal of this part of the proof is to show that, for mistake tree of $\cX, \cC$ of depth slightly less than $rk$, there exists a leaf $s$ such that the corresponding path from root to $s$ indicates an inclusion of a large non-repetitive set; in our notation, this means that we would like to identify a leaf $s$ such that $IJ(\rho^{-1}_{\cT, s}(1))$ is large. Since we will also need a similar bound later in the proof, we will prove the following lemma, which is a generalization of the stated goal that works even for the concept class $\cC[\Snr]$ for any non-repetitive $\Snr$. To get back the desired bound, we can simply set $\Snr = \emptyset$. 

\begin{lemma} \label{lem:good-leaves-comp}
For any non-repetitive set $\Snr$ and any depth-$d$ mistake tree $\cT$ of $\cX, \cC[\Snr]$, there exists a leaf $s \in \{0, 1\}^d$ such that $|IJ(\rho^{-1}_{\cT, s}(1)) \setminus IJ(\Snr)| \geqs d - r$.
\end{lemma}

The proof of this lemma is a double counting argument where we count a specific class of leaves in two ways, which ultimately leads to the above bound. The leaves that we focus on are the leaves $s \in \{0, 1\}^d$ such that, for every $(i, j)$ such that an $(i, j)$-assignment element appears in the path from root to $s$ but not in $\Snr$, the first appearance of $(i, j)$-assignment element in the path is included. In other words, for every $(i, j) \in IJ(P_{\cT, s}) \setminus IJ(\Snr)$, if we define $u_{i, j} \triangleq \inf_{s' \in \pre(s), v_{\cT, s'} \in \cX_{i, j}} |s'|$, then $s_{u_{i, j} + 1}$ must be equal to $1$. We call these leaves the \emph{good} leaves. Denote the set of good leaves of $\cT$ by $\cG_{\cT, \Snr}$.

Our first way of counting is the following lemma. Informally, it asserts that different good leaves agree with different sets $\tH \subseteq [r]$. This can be thought of as an analogue of Lemma~\ref{lem:soundness-mult} in our proof for VC Dimension. Note that this lemma immediately gives an upper bound of $2^r$ on $|\cG_{\cT, \Snr}|$.

\begin{lemma} \label{lem:good-leaves-upper}
For any depth-$d$ mistake tree $\cT$ of $\cX, \cC[\Snr]$ and any different good leaves $s_1, s_2 \in \cG_{\cT, \Snr}$, if $C_{I_1, H_1, \sigma_1}$ agrees with $s_1$ and $C_{I_2, H_2, \sigma_2}$ agrees with $s_2$ for some $I_1, I_2, H_1, H_2, \sigma_1, \sigma_2$, then $\tau(H_1) \ne \tau(H_2)$.
\end{lemma}

\begin{proof}
Suppose for the sake of contradiction that there exist $s_1 \ne s_2 \in \cG_{\cT, \Snr}$, $H_1, H_2, I_1, I_2, \sigma_1, \sigma_2$ such that $C_{I_1, H_1, \sigma_1}$ and $C_{I_2, H_2, \sigma_2}$ agree with $s_1$ and $s_2$ respectively, and $\tau(H_1) = \tau(H_2)$. Let $s$ be the common ancestor of $s_1, s_2$, i.e., $s$ is the longest string in $\pre(s_1) \cap \pre(s_2)$. Assume w.l.o.g. that $(s_1)_{|s| + 1} = 0$ and $(s_2)_{|s| + 1} = 1$.
Consider the node $v_{\cT, s}$ in tree $\cT$ where the paths to $s_1, s_2$ split; suppose that this is $x_{i, \sigma_i, j}$. Therefore $x_{i, \sigma_i, j} \in C_{I_2, H_2, \sigma_2} \setminus C_{I_1, H_1, \sigma_1}$.

We now argue that there is some $x_{i, \sigma'_i, j}$ (with the same $i,j$ but a different assignment $\sigma'_i$) that is in both concepts, i.e. $x_{i, \sigma'_i, j} \in C_{I_2, H_2, \sigma_2} \cap C_{I_1, H_1, \sigma_1}$. We do this by considering two cases:
\begin{itemize}
\item
If $(i, j) \in IJ(\Snr)$, then there is $x_{i, \sigma'_i, j} \in \Snr \subseteq C_{I_1, H_1, \sigma_1}, C_{I_2, H_2, \sigma_2}$ for some $\sigma'_i \in \Sigma^{U_i}$. %Since $\Snr $, both concepts contain $x_{i, \sigma'_i, j}$. 
\item 
Suppose that $(i, j) \notin IJ(\Snr)$. Since $s_1$ is a good leaf, there is some $t \in \pre(s)$ such that $v_{\cT, t} = x_{i, \sigma'_i, j}$ for some $\sigma'_i \in \Sigma^{U_i}$ and $t$ is included by the path (i.e. $s_{|t| + 1} = 1$). This also implies that $x_{i, \sigma'_i, j}$ is in both $C_{I_1, H_1, \sigma_1}$ and $C_{I_2, H_2, \sigma_2}$. 
\end{itemize}
%Hence, in both cases, there exists $\sigma'_i \in \Sigma^{U_i}$ such that $x_{i, \sigma'_i, j} \in C_{I_1, H_1, \sigma_1} \cap C_{I_2, H_2, \sigma_2}$.

Now, since both $x_{i, \sigma_i, j}$ and $x_{i, \sigma'_i, j}$ are in the concept $C_{I_2, H_2, \sigma_2}$, we have $(i, j) \in I_2$ and
\begin{align} \label{eq:sig-all-eq}
\sigma_i|_{\cN_i(M_{\tau(H_1)})} = \sigma_2|_{\cN_i(M_{\tau(H_1)})} = \sigma'_i|_{\cN_i(M_{\tau(H_1)})}.
\end{align}
On the other hand, since $C_{I_1, H_1, \sigma_1}$ contains $x_{i, \sigma'_i, j}$ but not $x_{i, \sigma_i, j}$, we have $(i, j) \in I_1$ and
\begin{align} %\label{eq:sig-all-eq}
\sigma_i|_{\cN_i(M_{\tau(H_2)})} \ne \sigma_1|_{\cN_i(M_{\tau(H_2)})} = \sigma'_i|_{\cN_i(M_{\tau(H_2)})}.
\end{align}
which contradicts \eqref{eq:sig-all-eq} since $\tau(H_1) = \tau(H_2)$.
\end{proof}

Next, we will present another counting argument which gives a lower bound on the number of good leaves, which, together with Lemma~\ref{lem:good-leaves-upper}, yields the desired bound.

\begin{proofof}[Lemma~\ref{lem:good-leaves-comp}]
For any depth-$d$ mistake tree $\cT$ of $\cC[\Snr], \cX$, let us consider the following procedure which recursively assigns a weight $\lambda_s$ to each node $s$ in the tree. At the end of the procedure, all the weight will be propagated from the root to good leaves.
\begin{enumerate}
\item For every non-root node $s \in \{0, 1\}^{\geqs 1}$, set $\lambda_s \leftarrow 0$. For root $s = \emptyset$, let $\lambda_{\emptyset} \leftarrow 2^d$.
\item While there is an internal node $s \in \{0, 1\}^{< d}$ such that $\lambda_s > 0$, do the following:
\begin{enumerate}
\item Suppose that $v_s = x_{i, \sigma_i, j}$ for some $i \in [r], \sigma_i \in \Sigma^{U_i}$ and $j \in [k]$.
\item If so far no $(i, j)$-element has appeared in the path or in $\Snr$, i.e., $(i, j) \notin IJ(P_{\cT, s}) \cup IJ(\Snr)$, then $\lambda_{s1} \leftarrow \lambda_s$. Otherwise, set $\lambda_{s0} = \lambda_{s1} = \lambda_s/2$.
\item Set $\lambda_s \leftarrow 0$.
\end{enumerate}
\end{enumerate}

The following observations are immediate from the construction:
\begin{itemize}
\item The total of $\lambda$'s over all the tree, $\sum_{s \in \{0, 1\}^{\leqs d}} \lambda_d$ always remain $2^d$.
\item At the end of the procedure, for every $s \in \{0, 1\}^{\leqs d}$, $\lambda_s \ne 0$ if and only if $s \in \cG_{\cT, \Snr}$.
\item If $s \in \cG_{\cT, \Snr}$, then $\lambda_s = 2^{|IJ(\rho^{-1}_{\cT, s}(1)) \setminus IJ(\Snr)|}$ at the end of the execution.
\end{itemize}
Note that the last observation comes from the fact that $\lambda$ always get divides in half when moving down one level of the tree unless we encounter an $(i, j)$-assignment element for some $i, j$ that never appears in the path or in $\Snr$ before. For any good leaf $s$, the set of such $(i, j)$ is exactly the set $IJ(\rho^{-1}_{\cT, s}(1)) \setminus IJ(\Snr)$.

As a result, we have $2^d = \sum_{s \in \cG_{\cT, \Snr}} 2^{|IJ(\rho^{-1}_{\cT, s}(1)) \setminus IJ(\Snr)|}$. Since Lemma~\ref{lem:good-leaves-upper} implies that $|\cG_{\cT, \Snr}| \leqs 2^r$, we can conclude that there exists $s \in \cG_{\cT, \Snr}$ such that $|IJ(\rho^{-1}_{\cT, s}(1)) \setminus IJ(\Snr)| \geqs d - r$ as desired.
\end{proofof}

\subsubsection{Part III: No Large Non-Repetitive Set Passes Many Test}

The main lemma of this subsection is the following, which is analogous to Lemma~\ref{lem:passing-bound}

\begin{lemma} \label{lem:non-pass}
If $\val(\cL) \leqs 0.001$, then, with high probability, for every non-repetitive set $\Snr$ of size at least $0.99rk$, $|\cH(\Snr)| \leqs 100n \log |\Sigma|$.
\end{lemma}

\begin{proof}
For every $I \subseteq [r]$, let $U_I \triangleq \bigcup_{i \in I} U_i$. For every $\sigma_I \in \Sigma^{U_I}$ and every $\tH \subseteq \cY$, we say that $(I, \sigma_I)$ \emph{passes} $\tH$ if $\sigma_I$ does not violate any constraint in $T_{\tH}$. Note that this definition and the way the test is generated in the reduction is the same as that of the VC Dimension reduction. Hence, we can apply Lemma~\ref{lem:passing-bound-equiv-form} with $\delta = 0.99$, which implies the following: with high probability,  for every $I \subseteq [r]$ of size at least $0.99 r$ and every $\sigma_I \in \Sigma^{U_I}$, $|\cH(I, \sigma_I)| \leqs 100n \log |\Sigma|$ where $\cH(I, \sigma_I)$ denote the set of all $\cH$'s passed by $(I, \sigma_I)$. Conditioned on this event happening, we will show that, for every non-repetitive set $\Snr$ of size at least $0.99 rk$, $|\cH(\Snr)| \leqs 100n \log |\Sigma|$.

Consider any non-repetitive set $\Snr$ of size $0.99rk$. Let $\sigma_{I(\Snr)}$ be an assignment on $U_{I(\Snr)}$ such that, for each $i \in I(\Snr)$, we pick one $x_{i, \sigma_i, j} \in \Snr$ (if there are more than one such $x$'s, pick one arbitrarily) and let $\sigma_{I(\Snr)}|_{U_i} = \sigma_i$. It is obvious that $\cH(\Snr) \subseteq \cH(I(\Snr), \sigma_{I(\Snr)})$. Since $\Snr$ is non-repetitive and of size at least $0.99rk$, we have $|I(\Snr)| \geqs 0.99r$, which means that $|\cH(I(\Snr), \sigma_{I(\Snr)})| \leqs 100n \log |\Sigma|$ as desired.
\end{proof}

\subsubsection{Part IV: A Subtree Containing $\Snr$ Must be Shallow}

In this part, we will show that, if we restrict ourselves to only concepts that contain some non-repetitive set $\Snr$ that passes few tests, then the Littlestone's Dimension of this restrictied concept class is small. Therefore when we build a tree for the whole concept class $\cC$, if a path from root to some node indicates an inclusion of a non-repetitive set that passes few tests, then the subtree rooted at this node must be shallow. %The specific bound we will show is stated below.

\begin{lemma} \label{lem:non-rep-shallow}
For every non-repetitive set $\Snr$, 
$$\LS(\cC[\Snr], \cU) \leqs 1.75rk - |\Snr| + r + 1000k\sqrt{r}\log(|\cH(\Snr)| + 1).$$
\end{lemma}

We prove the above lemma by bounding $\LS(\cC[\Snr], \cX)$ and $\LS(\cC[\Snr], \cY)$ separately, and combining them via Fact~\ref{fact:ls-union}. First, we can bound $\LS(\cC[\Snr], \cX)$ easily by applying Lemma~\ref{lem:good-leaves-comp} coupled with the fact that $|IJ(\Snr)| = |\Snr|$ for every non-repetitive $\Snr$. This immediately gives the following corollary.

\begin{corollary} \label{cor:non-rep-shallow-x}
For every non-repetitive set $\Snr$, 
$$\LS(\cC[\Snr], \cX) \leqs rk - |\Snr| + r.$$
\end{corollary}

We will next prove the following bound on $\LS(\cC[\Snr], \cY)$. Note that Corollary~\ref{cor:non-rep-shallow-x}, Lemma~\ref{lem:non-rep-shallow-y}, and Fact~\ref{fact:ls-union} immediately imply Lemma~\ref{lem:non-rep-shallow}.

\begin{lemma} \label{lem:non-rep-shallow-y}
For every non-repetitive set $\Snr$, 
$$\LS(\cC[\Snr], \cY) \leqs 0.75rk + 500k \sqrt{r} \log(|\cH(\Snr)| + 1).$$
\end{lemma}

The overall outline of the proof of Lemma~\ref{lem:non-rep-shallow-y} is that we will design a prediction algorithm whose mistake bound is at most $0.75rk + 1000k \sqrt{r} \log |\cH(\Snr)|$. Once we design this algorithm, Lemma~\ref{lem:learn-ls} immediately implies Lemma~\ref{lem:non-rep-shallow-y}. To define our algorithm, we will need the following lemma, which is a general statement that says that, for a small collection of $H$'s, there is a some $\tH^* \subseteq [r]$ that agrees with almost half of every $H$ in the collection.

\begin{lemma} \label{lem:half-agr}
Let $\cH \subseteq \cP([r])$ be any collections of subsets of $[r]$, there exists $\tH^* \subseteq [r]$ such that, for every $\tH \in \cH$, $|\tH^* \Delta \tH| \leqs 0.5r + 1000\sqrt{r}\log(|\cH| + 1)$ where $\Delta$ denotes the symmetric difference between two sets.
\end{lemma}

\begin{proof}
We use a simple probabilistic method to prove this lemma. Let $\tH^r$ be a random subset of $[r]$ (i.e. each $i \in [r]$ is included independently with probability 0.5). We will show that, with non-zero probability, $|\tH^r \Delta \tH| \leqs 0.5r + 1000\sqrt{r}\log(|\cH| + 1)$ for all $\tH \in \cH$, which immediately implies that a desired $\tH^*$ exists.

Fix $\tH \in \cH$. Observe that $|\tH^r \Delta \tH|$ can be written as $\sum_{i \in [r]} \mathds{1}[i \in (\tH^r \Delta \tH)]$. For each $i$, $\mathds{1}[i \in (\tH^r \Delta \tH)]$ is a $0, 1$ random variable with mean 0.5 independent of other $i' \in [r]$. Applying Chernoff bound here yields
\begin{align*}
\Pr[|\tH^r \Delta \tH| > 0.5r + 1000\sqrt{r}\log(|\cH| + 1)] \leqs 2^{-\log^2(|\cH| + 1)} \leqs \frac{1}{|\cH| + 1}. 
\end{align*}

Hence, by union bound, we have
\begin{align*}
\Pr[\exists \tH \in \cH, |\tH^r \Delta \tH| > 0.5r + 1000\sqrt{r}\log(|\cH| + 1)] \leqs \frac{|\cH|}{|\cH| + 1} < 1.
\end{align*}
In other words, $|\tH^r \Delta \tH| \leqs 0.5r + 1000\sqrt{r}\log(|\cH| + 1)$ for all $\tH \in \cH$ with non-zero probability as desired.
\end{proof}

We also need the following observation, which is an analogue of Observation~\ref{obs:pass-non-rep} in the VC Dimension proof; it follows immediately from definition of $\cH(S)$.

\begin{observation} \label{obs:pass-non-rep-ls}
If a non-repetitive set $\Snr$ is a subset of some concept $C_{I, H, \sigma_{\tau(H)}}$, then $\tau(H) \in \cH(\Snr)$.
\end{observation}

With Lemma~\ref{lem:half-agr} and Observation~\ref{obs:pass-non-rep-ls} in place, we are now ready to prove Lemma~\ref{lem:non-rep-shallow-y}.

\begin{proofof}[Lemma~\ref{lem:non-rep-shallow-y}]
Let $\tH^* \subseteq [r]$ be the set guaranteed by applying Lemma~\ref{lem:half-agr} with $\cH = \cH(\Snr)$. Let $H^* \triangleq \tH^* \times [k]$.

Our prediction algorithm will be very simple: it always predicts according to $H^*$; i.e., on an input\footnote{We assume w.l.o.g. that input elements are distinct; if an element appears multiple times, we know the correct answer from its first appearance and can always correctly predict it afterwards.} $y \in \cY$, it outputs $\mathds{1}[y \in H^*]$. Consider any sequence $(y_1, h_1), \dots, (y_w, h_w)$ that agrees with a concept $C_{I, H, \sigma_{\tau(H)}} \in \cC[\Snr]$. Observe that the number of incorrect predictions of our algorithm is at most $|H^* \Delta H|$.

Since $C_{I, H, \sigma_{\tau(H)}} \in \cC[\Snr]$, Observation~\ref{obs:pass-non-rep-ls} implies that $\tau(H) \in \cH(\Snr)$. This means that $|\tau(H) \Delta \tH^*| \leqs 0.5r + 1000\sqrt{r}\log(|\cH| + 1)$. Now, let us consider each $i \in [r] \setminus (\tau(H) \Delta \tH^*)$. Suppose that $i \in \tau(H) \cap \tH^*$. Since $i \in \tau(H)$, at least $k/2$ elements of $\cY_i$ are in $H$ and, since $i \in \tH^*$, we have $\cY_i \subseteq H^*$. This implies that $|(H^* \Delta H) \cap Y_i| \leqs k/2$. A similar bound can also be derived when $i \notin \tau(H) \cap \tH^*$. As a result, we have
\begin{align*}
|H^* \Delta H|
&= \sum_{i \in [r]} |(H^* \Delta H) \cap Y_i| \\ 
&= \sum_{i \in \tau(H) \Delta \tH^*} |(H^* \Delta H) \cap Y_i| + \sum_{i \in [r] \setminus (\tau(H) \Delta \tH^*)} |(H^* \Delta H) \cap Y_i| \\
&\leqs (|\tau(H) \Delta \tH^*|)(k) + (r - |\tau(H) \Delta \tH^*|)(k/2) \\
&\leqs 0.75rk + 500k\sqrt{r}\log(|\cH| + 1),
\end{align*}
concluding our proof of Lemma~\ref{lem:non-rep-shallow-y}.
\end{proofof}

\subsubsection{Putting Things Together}

\begin{proofof}[Lemma~\ref{lem:soundness-ls}]
Assume that $\val(\cL) \leqs 0.001$. From Lemma~\ref{lem:non-pass}, we know that, with high probability, $|\cH(\Snr)| \leqs 100n\log|\Sigma|$ for every non-repetitive set $\Snr$ of size at least $0.99rk$. Conditioned on this event, we will show that $\LS(\cC, \cU) \leqs 1.999rk$.

Suppose for the sake of contradiction that $\LS(\cC, \cU) > 1.999rk$. Consider any depth-$1.999rk$ mistake tree $\cT$ of $\cC, \cU$. From Lemma~\ref{lem:ls-assign-not-high}, no test-selection element is assigned to any node in the first $1.999rk - 1.001rk - 1 \geqs 0.997rk$ levels. In other words, the tree induced by the first $0.997rk$ levels is simply a mistake tree of $\cC, \cX$. By Lemma~\ref{lem:good-leaves-comp} with $\Snr = \emptyset$, there exists $s \in \{0, 1\}^{0.997rk}$ such that $|IJ(\rho_{\cT, s}^{-1}(1))| \geqs 0.997rk - r \geqs 0.996rk$.

Since $|IJ(\rho_{\cT, s}^{-1}(1))| \geqs 0.996rk$, there exists a non-repetitive set $\Snr \subseteq \rho_{\cT, s}^{-1}(1)$ of size $0.996rk$. Consider the subtree rooted at $s$. This is a mistake tree of $\cC[\rho_{\cT, s}], \cU$ of depth $1.002rk$. Since $\Snr \subseteq \rho_{\cT, s}^{-1}(1)$, we have $\cC[\rho_{\cT, s}] \subseteq \cC[\Snr]$. However, this implies
\begin{align*}
1.002rk &\leqs \LS(\cC[\rho_{\cT, s}], \cU) \\
&\leqs \LS(\cC[\Snr], \cU) \\
(\text{From Lemma~\ref{lem:non-rep-shallow}}) &\leqs 1.75rk - 0.996rk + r + 100k\sqrt{r}\log(|\cH(\Snr)| + 1) \\
(\text{From Lemma~\ref{lem:non-pass}}) &\leqs 0.754rk + r + 100k\sqrt{r}\log(100n\log|\Sigma| + 1) \\
&= 0.754rk + o(rk),
\end{align*}
which is a contradiction when $r$ is sufficiently large.
\end{proofof}

\section{Conclusion and Open Questions}
% !TeX root = main.tex

In this work, we prove inapproximability results for VC Dimension and Littlestone's Dimension based on the randomized exponential time hypothesis. Our results provide an almost matching running time lower bound of $n^{\log^{1 - o(1)} n}$ for both problems while ruling out approximation ratios of $1/2 + o(1)$ and $1 - \varepsilon$ for some $\varepsilon > 0$ for VC Dimension and Littlestone's Dimension respectively.  Even though our results help us gain more insights on approximability of both problems, it is not yet completely resolved. More specifically, we are not aware of any constant factor $n^{o(\log n)}$-time approximation algorithm for either problem; it is an intriguing open question whether such algorithm exists and, if not, whether our reduction can be extended to rule out such algorithm. Another potentially interesting research direction is to derandomize our construction; note that the only place in the proof in which the randomness is used is in Lemma~\ref{lem:passing-bound-equiv-form}.

A related question which remains open, originally posed by Ben-David and Eiron~\cite{BE98-self_directed}, is that of computing the {\em self-directed learning}\footnote{Roughly, self-directed learning is similar to the online learning model corresponding to Littlestone's dimension, but where the learner chooses the order elements; see~\cite{BE98-self_directed} for details.} mistake bound. 
Similarly, it may be interesting to understand the complexity of computing (approximating) the recursive teaching dimension~\cite{DFSZ14-RT, MSWY15-RT}.

\section*{Acknowledgement}
We thank Shai Ben-David for suggesting the question of approximability of Littlestone's dimension, and several other fascinating discussions. We also thank Yishay Mansour and COLT anonymous reviewers for their useful comments. 

Pasin Manurangsi is supported by NSF Grants No. CCF 1540685 and CCF 1655215.

Aviad Rubinstein was supported by a Microsoft Research PhD Fellowship, as well as NSF grant CCF1408635 and Templeton Foundation grant 3966. This work was done in part at the Simons Institute for the Theory of Computing.

\bibliographystyle{alpha}
\bibliography{main}

\appendix

\section{Quasi-polynomial Algorithm for Littlestone's Dimension} \label{app:littlestone-alg}
% !TeX root = main.tex

In this section, we provides the following algorithm which decides whether $\LS(\cC, \cU) \leqs d$ in time $O(|\cC| \cdot (2|\cU|)^d)$. Since we know that $\LS(\cC, \cU) \leqs \log |\cC|$, we can run this algorithm for all $d \leqs \log |\cC|$ and compute Littlestone's Dimension of $\cC, \cU$ in quasi-polynomial time.

\begin{theorem}[Quasi-polynomial Time Algorithm for Littlestone's Dimension]\label{thm:LS-alg}
There is an algorithm that, given a universe $\cU$, a concept class $\cC$ and a non-negative integer $d$, decides whether $\LS(\cC, \cU) \leqs d$ in time $O(|\cC| \cdot (2|\cU|)^{d})$.
\end{theorem}

\begin{proof}
Our algorithm is based on a simple observation: if an element $x$ belongs to at least one concept and does not belong to at least one concept, the maximum depth of mistake trees rooted at $x$ is exactly 
$1+\min\left\{\LS(\cC[x \to 0], \cU), \LS(\cC[x \to 1], \cU)\right\}$. Recall from Section~\ref{sec:littlestone} that $\cC[x \to 0]$ and $\cC[x \to 1]$ denote the collection of concepts that exclude $x$ and the collection of concepts that include $x$ respectively.

This yields the following natural recursive algorithm. For each $x \in \cU$ such that $\cC[x \to 0], \cC[x \to 1] \ne \emptyset$, recursively run the algorithm on $(\cC[x \to 0], \cU, d - 1)$ and $(\cC[x \to 1], \cU, d - 1)$. If both executions return NO for some $x$, then output NO. Otherwise, output YES. When $d = 0$, there is no need for recursion as we can just check whether $|\cC| \leqs 1$.

Finally, we note that the running time can be easily proved by induction on $d$. 
\end{proof}
\end{document}